\documentclass{article}

\usepackage{authblk}
\usepackage{amssymb,amsmath,amsthm}

\usepackage{psfig,epsfig}
\newtheorem{theorem}{Theorem}[section]
\newtheorem{lemma}[theorem]{Lemma}
\newtheorem{claim}[theorem]{Claim}

\newcommand{\comment}[1]{}
\newcommand{\NP}{\mathbb{NP}}
\newcommand{\R}{\mathbb{R}}

\newcommand{\Xomit}[1]{}

\newcommand{\E}{\mbox{\bf E}}

\begin{document}

% Page heads
%\markboth{K. Andreev et al.}{Algorithms for Constructing Overlay Networks For Live Streaming}

% Title portion
\title{Algorithms for Constructing Overlay Networks For Live Streaming}

\author[1]{KONSTANTIN ANDREEV}
\author[2,6]{BRUCE M. MAGGS}
\author[3]{ADAM MEYERSON}
\author[4]{JEVAN SAKS}
\author[5,6]{RAMESH K. SITARAMAN}
\affil[1] {Oppenheimer Funds
Two World Financial Center, 225 Liberty Street, 11th Floor
New York, NY 10281, Email: {\tt kandreev@oppenheimerfunds.com}}
\affil[2]{Department of Computer Science, Duke University,
Durham, NC 27708-0129, Email: {\tt bmm@cs.duke.edu}}
\affil[3]{Department of Computer Science, 4732 Boelter Hall,   University of California, Los Angeles, CA 90095, Email: {\tt awm@cs.ucla.edu}}
\affil[4]{Microsoft, Seattle, WA}
\affil[5]{Department of Computer Science, University of Massachusetts,  Amherst MA 01003, Email: {\tt ramesh@cs.umass.edu}}
\affil[6]{Akamai Technologies, 8 Cambridge Center, Cambridge, MA 02142}
\maketitle
\begin{abstract}
In this paper, we present a polynomial time approximation algorithm for constructing an overlay multicast network for streaming live media events over the Internet.  The class of overlay networks constructed by our algorithm include networks used by Akamai Technologies to deliver live media events to a global audience with high fidelity. In particular, we construct networks consisting of three stages of nodes.  The nodes in the first stage are the entry points that act as sources for the live streams.  Each source forwards each of its streams to one or
more nodes in the second stage that are called reflectors.  A
reflector can split an incoming stream into multiple identical
outgoing streams, which are then sent on to nodes in the third and
final stage that act as sinks and are located in edge networks near end-users.  As the packets in a stream
travel from one stage to the next, some of them may be lost.  The job
of a sink is to combine the packets from multiple instances of the
same stream (by reordering packets and discarding duplicates) to form
a single instance of the stream with minimal loss.  We assume that the
loss rate between any pair of nodes in the network is known, and that
losses between different pairs are independent, but discuss an extension
to tolerate failures that happen in a coordinated fashion. Our primary contribution is an algorithm that constructs an overlay network that provably
satisfies capacity and reliability constraints to within a constant
factor of optimal, and minimizes cost to within a logarithmic factor of optimal.  Further in the common case where only the transmission costs are minimized, we show that our algorithm produces a solution that   has cost within a factor of 2 of optimal. We also implement our algorithm and evaluate it on realistic traces derived from Akamai's live streaming network. Our empirical results show that our algorithm can be used to efficiently construct large-scale overlay networks in practice with near-optimal cost. 
\end{abstract}

%\category{F.2.2}{Analysis of Algorithms and Problem Complexity}{Nonnumerical Algorithms and Problems}[Routing and layout]
%\category{C.2.1}{Computer-Communication Networks}{Network Architecture and Design }[Distributed Networks]
%\category{C.2.4}{Computer-Communication Networks}{Distributed Systems}[Distributed Applications]

%\terms{Algorithms, Theory, Design, Reliability.}

%\keywords{Approximation Algorithms, Network Design, Streaming Media, Content Delivery, Network Reliability}

%\acmformat{}

\section{Introduction}
One of the most appealing applications of the Internet is the
delivery of high-quality live video streams to the end-user's 
desktop or device at low cost. Live streaming is becoming increasingly
popular as more and more enterprises want to stream on the
Internet to reach a world-wide audience. Common examples include radio and television broadcasts, live events with a global viewership, sporting events, and multimedia conferencing. As all forms of traditional media inexorably migrate to the Internet, there has been sea-change in recent years in what is expected from live streaming technology. Today, broadcasters and end-users increasingly expect a high-quality live viewing experience that is nearly loss-free with fidelity comparable to that of a high-definition (HD) television broadcast!  A promising technology for delivering live streams with high fidelity is building an overlay network that can ``mask'' the packet loss and failures inherent in the Internet by using replication and redundancy \cite{KontothanasisSWHKMSS04}. Algorithms that automatically construct such overlay networks are a key ingredient of overlay streaming technology \cite{NygrenSS10}.  Such algorithms need to be efficient, since the failure and loss characteristics of the Internet change frequently, necessitating the periodic reconstruction of the overlay network. Further, the cost of delivering the streams using the overlay network needs to be minimized, so as to make the overall cost of live online media affordable. The primary contribution of this work is formulating the overlay network construction problem for live streams and developing efficient algorithms that construct overlay networks that provably provide high quality service at low operating cost.  

It is instructive to contrast overlay streaming technology  with the traditional approach to live streaming. The traditional centralized approach to delivering live streaming
involves three steps. First, the event is captured and encoded
using an {\it encoder\/}. Next, the encoder delivers the encoded
data to one more {\it media servers\/} housed in a centralized
co-location center\footnote{A co-location center is a data center that provides power, rack space, and Internet connectivity for hosting a large number of servers.} on the Internet. Then, the media server
streams the data to a media player on the end-user's computer.
Significant advances in encoding technology, such as MPEG-2 and H.264/MPEG-4, have
made it possible to achieve full-screen High-Definition (HD) television quality video with data rates between 2 to 20 megabits per second.  However, transporting the streaming bits across the Internet from the
encoder to the end-user without a significant loss in stream
quality is a critical problem that is hard to resolve with the traditional approach. More specifically, the traditional centralized approach for stream delivery outlined
above has two bottlenecks, both of which argue for the
construction of an overlay  network for delivering
live streams.

\noindent{\bf Server bottleneck.} Most media servers
can serve no more than several hundred Mbps of streams to end-users. In January
2009, Akamai hosted President Obama's inauguration event 
which drew 7 million simultaneous viewers world-wide with a peak aggregate traffic of 2 Terabits per second (Tbps). Demand for viewing live streams continues to rise quickly, spurred by a continual increase in broadband speed and penetration rates \cite{Belson10}. In April 2010, Akamai hit a new record peak of 3.45 Tbps on its network. At this throughput, the entire printed contents of the U.S. Library of Congress could be downloaded in under a minute. In the near term (two to five years), it is reasonable to expect that throughput requirements for some single video events will reach roughly 50 to 100 Tbps (the equivalent of distributing a TV-quality stream to a large prime time audience). This is an order of magnitude larger than the biggest online events today. To host an event of this magnitude 
requires tens of thousands of servers. In addition these servers must be
distributed across multiple co-location centers, since few co-location
centers can provide even a fraction of the required outgoing bandwidth  to end-users. 
Furthermore, a
single co-location center is a single point of failure. Therefore,
scalability and reliability requirements dictate the need for a
distributed infrastructure consisting of a large number of servers 
deployed across the Internet.

\noindent{\bf Network bottleneck.} As live media is increasingly
streamed to a global viewership, streaming data needs to be
transported reliably and in real-time from the encoder to the
end-user's media player over the long haul across the Internet.
The Internet is designed as a best-effort network with no quality
guarantees for communication between two end points, and packets
can be lost or delayed as they pass through congested routers or
links. This can cause the stream to degrade, producing ``glitches'',
``slide-shows'', and ``freeze ups'' as the user watches the stream. In
addition to degradations caused by packet loss,
catastrophic events can cause complete denial of service to segments of the audience. These
events include complete failure of large 
{\em Internet Service Providers\/} (ISP's), or failing of ISP's to peer with
each other. As an example of the former, in January 2008 an undersea cable cut brought down networks in the Middle East and India, dramatically impacting Internet services for several hours and taking several days to return to normality.
As an example of the latter, in June 2001, Cable and Wireless 
abruptly stopped peering with PSINet for financial reasons.
In the traditional centralized delivery model, it is customary to
ensure that the encoder is able to communicate well with the media
servers through a dedicated leased line, a satellite uplink, or
through co-location. However, delivery of bits from the media
servers to the end-user over the long haul is left to the vagaries
of the Internet.

The bottlenecks outlined above speak to the need of a distributed overlay network for delivering live streams. For more comprehensive treatment of the Internet bottlenecks and the architecture of delivery networks in general, including media and application delivery, the reader is referred to \cite{NygrenSS10}. 

\subsection{An overlay network for delivering live streams}

The purpose of an overlay network is to transport each stream from its
encoder to its viewers in a manner that alleviates the server and
network bottlenecks. An {\em overlay network\/}  can be represented by a tripartite digraph $N = (V, E)$, and a set of paths $\Pi$ in $N$ that are used to transport the streams (See Figure~\ref{tripartite}).  The node set $V$ consists of a set of sources $S$ representing entry points, a set $R$ representing reflectors, and a set of sinks $D$ representing edge servers. Physically, each node is a {\em cluster\/} of machines deployed within a data center of an ISP on the Internet. The nodes are globally distributed and are located in diverse ISPs and geographies across the Internet.  The set of edges $E = (S \times R) \cup (R \times D)$ denote links that can potentially be used for transporting the streams.  Note that transporting a stream across a link $(u,v) \in E$ involves a server at node $u$ sends a sequence of packets that constitutes the stream to a server at $v$ using the public Internet.

Given the server deployments that are represented by $N$, overlay network construction entails computing the set of paths $\Pi$ that can be used to route each stream from its source to each of its sinks. Each path in $\Pi$ originates at a source, passes through a reflector, and terminates at a sink. Note that there can be more than one path between a source and sink when multiple copies need to be sent to enhance stream quality.
\begin{figure}[t]
        \centerline{\psfig{figure=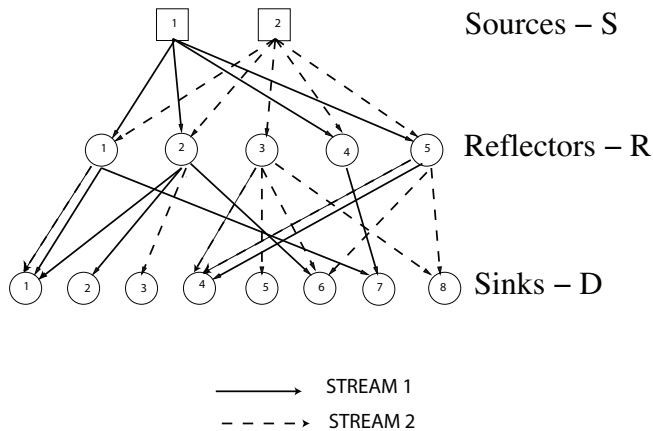,width=5in}}
        \caption{An overlay network for live streaming }
\label{tripartite}
\end{figure}
We illustrate the functionality of an overlay network by tracking the path of a stream through the
overlay network as it travels from the encoder to the end-user's
media player.
\begin{itemize}
\item An {\it entry point\/} (or, {\it source\/}) serves as the point of entry for the
stream into the overlay network and it receives the sequence of
packets that constitutes the stream from the encoder. The
entry point then sends identical copies of the stream to one or
more reflectors. For instance, in Figure~\ref{tripartite},  source  1 originates stream 1 and it forwards the stream to reflectors 1, 2, 4, and 5.

\item A {\it reflector\/} serves as a ``splitter'' and can send
each stream that it receives to one or more edge servers. For instance, in Figure~\ref{tripartite},  reflector 1 forwards stream 1 to sinks 1 and 7, in addition to forwarding stream 2 to sink 1.

\item An {\it edge server\/} (or {\it sink\/}) receives one or more identical copies
of the stream, each from a different reflector, and
``reconstructs'' a cleaner copy of the stream, before sending it to the
media player of the end-user. Specifically, if the $k^{th}$ packet
is missing in one copy of the stream, the edge server waits for
that packet to arrive in one of the other identical copies of the stream
and uses it to fill the ``hole''. For instance,  in Figure~\ref{tripartite}, stream 2 is sent from source 2 to sink 8 through two edge-disjoint paths, one through reflector 3 and the other through reflector 5. Any packet lost on the path through reflector 3 can be recovered if that same packet is not lost on the path through reflector 5. The process of stream replication and reconstruction is key to ensuring a high quality of service (QoS) when there is no single reliable loss-free path from the source to the sink.  If the packet is lost on all paths, then that packet is unrecoverable by the sink. The unrecoverable loss is termed as {\em end-to-end\/} packet loss or {\em post-reconstruction\/} packet loss.
\end{itemize}
The architecture of the overlay network allows for
distributing a stream from its entry point to a large number of
edge servers with the help of reflectors, thus alleviating the
server bottleneck. The network bottleneck can be broken down into
three parts. The first-mile bottleneck from the encoder to the
entry point can be alleviated by choosing an entry point close to
(or even co-located with) the encoding facility. The middle-mile
bottleneck of transporting bits over the long-haul from the
entry point to the edge servers can be alleviated by building an
overlay network that supports low loss and high reliability. This
is the hardest bottleneck to overcome, and algorithms for
automatically constructing such an overlay network is the primary contribution of this paper. The last-mile
bottleneck from the edge server to the end-user can be alleviated
to a significant degree by deploying edge servers ``close'' to end-users (in a network sense) and mapping each user to the most proximal edge server. Further, with significant growth of broadband
into the homes of end-users, the last-mile bottleneck is bound to
become less significant in the future\footnote{From the year 2000 to October 2009, the percentage US households with high-speed broadband Internet services grew from a mere 4.4\% to 63.5\%. Recent statistics for broadband penetration derived from Akamai data can be found in \cite{Belson10}.}.

\subsection{Considerations for overlay network construction}
Given  a digraph $N = (V, E)$ that represents a deployment of sources, reflectors, and sinks, and given a set of live streams,  the construction of an overlay network involves computing a set of paths $\Pi$ that specify how each stream is routed from its  source to the subset of the sinks that are
designated to serve that stream to end-users.  As an example, in Figure~\ref{tripartite}, we are given a set of 2 sources, 5 reflectors, 8 sinks, and 2 streams. Further, we are given the designated subset of sinks for stream 1 and 2 to be  $\{1, 2, 4, 6, 7\}$ and $\{1, 3, 4, 6, 8\}$ respectively. The goal of overlay construction is to create one or more paths from each source to each sink that requires the stream. 

In practice, the designated subset of sinks for a given stream 
takes into account the expected viewership of that stream. For instance, a large live event with predominantly European viewership would
include a large number of sinks (i.e., edge servers) in Europe in its designated subset, so as to provide many proximal choices to the viewers of that stream.  Constructing the designated subset of sinks for each stream and subsequently directing each viewer to his/her most proximal sink within that designated subset, so as  to alleviate the last-mile bottleneck is called ``mapping'' (For a more technical  discussion on mapping, see \cite{DilleyMPPSW02,NygrenSS10}.) Mapping is a complementary problem to overlay network construction and is not a topic of this paper. From the perspective of constructing an overlay network, the source of each stream and the corresponding subset
of the sinks that are designated to serve that stream are simply given to us as inputs to  our  algorithm. Note that a physical entry point deployment may originate multiple streams and a physical edge server deployment may typically receive and serve a number of distinct streams. However, for simplicity, and without loss of generality, we will replicate the sources (resp., sinks) so that each source (resp., sink) originates (resp., receives) exactly one stream.

As noted earlier, given the sources, reflectors, sinks, and streams, constructing an overlay network involves constructing paths $\Pi$ to transport each stream to its designated subset of sinks. Overlay network construction can be viewed as an optimization problem to minimize {\it cost\/}, subject to  {\it capacity\/}, {\it quality\/}, and {\it reliability\/} requirements as discussed below.

\noindent{\bf Cost:\/} A significant fraction of the cost of operating an overlay
network is the transmission cost of sending traffic over the Internet.
The sources, reflectors, and sinks are servers co-located
in  data centers of ISPs across the Internet. Operating the
overlay network requires entering into contracts with each
data center (typically, owned by an ISP) for bandwidth use in and out of the facility. A typical bandwidth contract is based either on average bandwidth use for the month, or on the $95^{th}$
percentile of five-minute-averages of the traffic for the month \cite{AdlerSV06}.
Therefore, depending on the specifics of the contract and 
usage in the month so far, it is possible to estimate the
incremental cost (in dollars) of sending additional bits across each link in
the overlay network.  We would like to minimize the total transmission cost of usage of all the links. While transmission costs represent a  major fraction of the operating costs of an overlay network, there are also fixed costs such as the amortized cost of procuring servers and the recurring co-location expenses. From the perspective of operating a streaming overlay network, these are often sunk costs that are often shared across services, with the possible exception of dedicated reflectors. Therefore, we do model a fixed cost for reflector usage.

\noindent{\bf Capacity:\/} The capacity constraints
reflect resource  and other limitations of the nodes and links in the overlay network. These constraints can be represented as capacities associated with the nodes and links of the digraph $N = (V,E)$. 
Bandwidth capacity specifies the maximum total bandwidth (in bits/sec) that can be sent by a given node or sent through a given link. Bandwidth capacity incorporates CPU,
memory, and other resource limitations on the physical hardware, as well as any 
bandwidth limitations for sending outbound traffic from a co-location facility.
For instance, a reflector machine may be able to push at most 100
Mbps before becoming CPU-bound. Another type of capacity that is also relevant is called fan-out and represents the maximum number of distinct streams a node (such as reflector) can transmit simultaneously.  It is critical to model fan-out constraints and bandwidth capacity for the reflectors, since reflector capacity is the key to scalably transmitting streams to large audiences. We start by modeling reflector fan-out constraints and later extend our solution to accommodate reflector bandwidth capacity in Section~\ref{sec:extensions}.  Note that, in addition to resource
limitations, one can also impose capacities to clamp down traffic
in certain locations and move traffic through other locations in the network for contractual  or cost reasons.

\noindent{\bf Quality:\/} The quality of the stream that an
edge server delivers to an end-user is directly related to how well the edge server is able to reconstruct the stream without incurring significant packet loss. Consequently, we would like to guarantee good stream {\em quality-of-service (QoS)\/} by requiring that the end-to-end packet loss for a stream at each relevant sink be no larger than a pre-specified loss threshold.

%Note that packets that arrive very late or significantly
%out-of-order must also be considered effectively useless, as they
%cannot be utilized in real-time for stream playback. 

\noindent{\bf Reliability:\/} From time to time, 
catastrophic events on the Internet
cause large segments of viewers to be denied service. To
defend against this possibility, the overlay network must be monitored and recomputed frequently to route the streams around major failures. In addition, one can place systematic constraints on overlay network construction to provide greater
fault-tolerance. An example of such a constraint is to require
that multiple copies of a given stream 
are always sent through reflectors
located in different ISPs. This constraint would protect 
against the catastrophic failure of any single
ISP. We explore incorporating such constraints into the construction of overlay networks in Section~\ref{sec:colconstraints}.

\subsubsection{Packet loss model}
To estimate the end-to-end (i.e., post-reconstruction) packet loss for a stream at an
edge server, we need to measure and model the packet loss of each
individual link in digraph $N = (V,E)$ of the overlay network. The packet loss on each
link can be periodically estimated by proactively sending test
packets to measure loss on that link. Typically, an exponentially-weighted historical average is used as an estimate of the packet loss on that link at a given point of time. In reality, losses in a link
tend to be bursty and correlate positively over time, i.e., if the
$k^{th}$ packet on a link was lost there is a slightly greater probability
that the $(k+1)^{st}$ packet on the same link will also be lost. Further, the loss
in different links can be correlated if the Internet routes corresponding to those links happen to pass through the same physical routers. However, as a first cut, it is quite reasonable to assume that
all loss events on two {\em different\/} links are independent and uncorrelated, i.e.,
losses on one link are completely unrelated to a losses on any other link. Notice that we don't 
assume that loss of packets on individual links are uncorrelated, but
we assume that losses on different links are independent\footnote{The assumption of loss independence between different links is not strictly true in practice, especially if the underlying Internet paths that these links represent share resources such as routers. However, we find the assumption to be a good first-cut approximation in practice that enables the design of efficient algorithms.}.
However, in Section~\ref{sec:extensions},  we explore extensions to this simplified loss model to account for some correlated link failures.

\subsubsection{Efficiency requirements}
An algorithm for constructing an overlay network needs to be efficient since  the inputs to our optimization change with time, requiring the  overlay network to be recomputed frequently. The digraph $N=(V,E)$ changes when new nodes are deployed, existing nodes fail and need to be taken out of service, or when streams are added or removed. The costs and capacities associated with the nodes and links also change over time, depending on server deployments,  bandwidth usage, and contracts.  For instance, suppose we have already used a reflector early in the month such that its
$95^{th}$ percentile of traffic for the month is guaranteed to be at least 40 Mbps. We can now set the capacity of that reflector to 40 Mbps and the link costs to zero and essentially use it for free for the rest of
month. Finally, the loss probabilities associated with each link in $E$ must be updated as loss conditions are measured on the Internet change. Thus, we need efficient algorithms that run in polynomial time and can feasibly (re)construct the overlay network, typically several times an hour. The need for algorithmic efficiency is particularly key since the size of the overlays are growing successively larger over time with growing streaming usage and wider deployments.

\subsection{Our contributions}

Our first key contribution is formulating the overlay network construction problem for live streams, an optimization problem that is at the heart of much of enterprise-quality live streaming today. Subsequently, we design the first  provably efficient algorithm for constructing an overlay network that obeys capacity and quality constraints while minimizing cost. We show that constructing the optimal overlay network is NP-Hard. However, we provide an approximation algorithm that constructs an overlay network that is provably near-optimal.  Specifically, we provide an efficient algorithm that constructs an overlay network that obeys all capacity and quality constraints to within a constant factor, and minimizes the cost to within $O(\log n)$ of optimal. The approximation bound for the cost is tight since set cover, which has a known logarithmic lower bound, is a special case of our problem. Further, for the important special case where only transmission costs are minimized, our algorithm produces a solution with cost that is provably within a factor of 2 of optimal. In addition, our algorithm can be extended to incorporate more complex constraints such as constructing  overlays that provide reliability even when individual ISPs fail.  Our technique is of independent interest and is based upon linear program rounding, combined with
a novel application of the generalized assignment algorithm \cite{ShmoysT93}. Finally, we implement our algorithm and 
show that it performs very well on a range of actual trace data collected from the Akamai's live streaming network. In particular, our algorithm produced near-optimal results within a feasible amount of time for real-world networks.

\subsection{Related work}

First, we discuss related work on systems that deliver streams utilizing multicast protocols, in lieu of the reflector-based overlay network that we study in our current work. Next, we discuss optimization research that is closely-related to our own algorithmic approach.

\subsubsection{Multicast protocols} 
One of the oldest alternative approaches to distributing streams is called ``multicast''~\cite{Deering91}.  The goal of multicast is to reduce the total
bandwidth consumption required to send the same stream to a large
number of hosts.  Instead of sending all of the data directly from one
server, a multicast tree is formed with a server at the root, routers
at the internal nodes, and end-users at the leaves.  A router receives
one copy of the stream from its parent and then forwards a copy to
each of its children.  The multicast tree is built automatically as
end-users subscribe to the stream.  The server does not keep track of
which end-users have subscribed.  It merely addresses all of the packets
in the stream to a special multicast address, and the routers take
care of forwarding the packets on to all of the end-users who have
subscribed to that address.  Support for multicast is provided at both the
network and link layer.  Special IP and hardware addresses have been
allotted to multicast, and many commercial routers support the
multicast protocols.

Unfortunately, few of the routers on major backbones are configured to
participate in the multicast protocols, so as a practical matter it is
not possible for a server to rely on multicast alone to deliver its
streams.  The ``Mbone'' (multicast backbone) network was organized to
address this problem~\cite{Eriksson94}.  Participants in Mbone have installed routers
that participate in the multicast protocols.  In Mbone, packets are
sent between multicast routers using unicast ``tunnels'' through
routers that do not participate in multicast.

However, multicast protocols have other issues as well. 
With the multicast protocols, the trees are not
very resilient to failures.  In particular, if a node or
link in a multicast tree fails, all of the leaves downstream of the
failure lose access to the stream.  While the multicast protocols do
provide for automatic reconfiguration of the tree in response to a
failure, end-users will experience a disruption while reconfiguration
takes place.  Similarly, if an individual packet is lost at a node or
link, all leaves downstream will see the same loss.  To compound
matters, the multicast protocols for building the tree do not attempt to minimize
end-to-end packet loss or maximize available bandwidth in the tree. In contrast, as noted earlier, commercial media delivery systems such as Akamai do not rely on multicast, but instead provide a new component called a reflector.  A reflector
receives one copy of a stream and then forwards one or more copies to
other reflectors or media servers.  Further, the media servers at the edge of the Internet are capable of receiving packets via multiple paths to recover from packet loss. 

The approach studied in our work involves constructing an overlay network in a centralized manner using dedicated hardware for entry points, reflectors, and edge servers. For a more detailed overview of the system architecture of Akamai's live streaming network, the reader is referred to \cite{KontothanasisSWHKMSS04}. And, for an analysis of the traffic on Akamai's live streaming network, the reader is referred to \cite{sripanidkulchai2004analysis}. A complementary approach is the peer-to-peer (P2P) approach to live streaming where end-user machines can self-organize themselves into an overlay tree that can be used to distribute media content. An example of such a system is ``End System Multicast''(ESM)~\cite{Chu00,LiuRLZ08}. In ESM, there is no distinction between reflectors, and edge servers.  Each host participating in the  multicast may be called on to play any of these roles simultaneously
in order to form a tree.  ESM allows multicast groups to be formed without any network support for routing protocols and without any other permanent infrastructure dedicated to supporting multicast.  While P2P live streaming is cost-effective, it is still unclear that it can provide the scale and quality-of-service of a dedicated overlay network. Recently, there has also been work on hybrid systems such as CoopNet~\cite{Ven01,PadmanabhanS02} that incorporate certain elements of both dedicated overlays and peer-to-peer systems. Still, the vast majority of enterprise-quality live streaming for key global events today happen using the dedicated overlay networks such as the one we study in this paper.

\subsubsection{Algorithms for facility location}

Our algorithmic approach is inspired by recent work on a general class of optimization problems known as facility location problems. In a classical version of the facility location problem, there are a set of potential locations where facilities may be built and a set of client locations that each require service from a facility. Given the cost of servicing each client from each facility location, the objective is to determine the set of locations at which facilities should be built so as to minimize the total cost of building the facilities and servicing all the clients. This class of problems has 
numerous applications in
operations research, databases, and computer networking. The first
approximation algorithm for facility location problems was given by
Hochbaum~\cite{Hochbaum82} and improved approximation algorithms have been the
subject of numerous papers including~\cite{Chudak98,CharikarG99,GuhaK98,JainMS02,JainV99,Mahdian,ShmoysTA97,Svir01}. Except for Hochbaum's result, the papers described above all assume
that the costs of  servicing clients from facilities form a metric
(satisfying the symmetry and triangle inequality properties). 
While our overlay network construction problem is significantly different from the prior work on facility location, we can provide a rough analogy where reflectors acts as facilities, sinks act as clients, and the costs are the
weights that represent packet loss probabilities. But the analogy does not fully capture the additional complexities and unique challenges that need to be overcome to solve our problem. Further,  the packet loss probabilities do not necessarily form a metric. And, the symmetry and triangle inequality constraints frequently fail in real networks.

Our overlay network construction problem includes set cover as a special case, though our
problem is much more general. But, it is instructive to review the set cover literature, since it provides lower bounds on the complexity of our problem. The fact that our problem has set
cover as a special case, gives us an approximation lower bound of $\Omega(\log n)$ 
with respect to cost achievable by a 
polynomial-time algorithm (unless $NP  \subset
DTIME(n^{O(\log\log{n})})$)~\cite{LundY94,Feige98}. A
simple greedy algorithm gives a matching upper bound for the set cover
problem~\cite{Johnson74,Chvatal79}.
Our problem is  capacitated (in
contrast to the set cover problem where the sets are
uncapacitated). Capacitated facility location (with ``hard''
capacities) has been considered in prior work \cite{PalTW01}, but the local
search algorithm provided depends heavily upon the use of an
underlying metric space. The standard greedy approach for the set
cover problem can be extended to accommodate capacitated sets, but our problem is significantly more complex as it requires a two-level assignment of sources (i.e., streams) to reflectors and reflectors to sinks. Two-level assignments have been considered previously in other contexts~\cite{KorupoluPR99,BaevR01,MeyersonMP01,GuhaM02},  though they assume that the points are located in a metric space. Another basic property of our problem that makes it less amenable to a greedy approach that has worked well in other contexts is that with multiple streams  the coverage no longer increases concavely as more reflectors are used to route the streams. In other words, using two additional reflectors may improve our solution by a larger margin than the sum of the improvements of the reflectors taken individually.

A unique feature of our overlay network construction problem is the ability to route a stream via multiple paths and combine the different copies at the sink to provide a high level of quality. This is reminiscent of extensions to the facility location problem where a given client is assigned to multiple facilities in a redundant fashion~\cite{JainV00,GuhaMM01}. However, unlike our results, each of the
previous papers assumes that the costs for connecting clients to facilities form a metric. Further, it is also assumed
that the coverage provided by each facility is equivalent (whereas in
our problem the reflectors provide benefit in a more complex manner by enabling multiple paths that decrease the end-to-end packet loss).

\subsubsection{Network reliability}
While our problem aims to construct an overlay network with low cost and providing a specified level of quality-of-service, there has been prior relevant work on network reliability that studies when a network becomes disconnected due to link failures. Given a network where each link $e$ fails (i.e., vanishes) independently with probability $p_e$, the all-terminal network reliability problem aims to determine the probability that the network becomes disconnected. Likewise, the two-terminal network reliability problem is to determine the probability that a specified source and sink node in the network are disconnected from each other.  Both problems and several related variations are known to be $\sharp$P-complete \cite{ProvanB83,Valiant79} for general networks. 
Karger, however,  showed a fully polynomial-time randomized approximation scheme (FPRAS) that approximates the all-terminal network reliability to within a relative error of $1 \pm \epsilon$ in time that is polynomial in the number of vertices and $1/\epsilon,$ with high probability~\cite{Karger99} . Further, Karger showed how his approach can be extended, with some restrictions, to a more general
problem, namely the multi-terminal network reliability problem, 
where instead of all terminals we have an arbitrary subset of terminals and we would like to compute the probability some pair of terminals in the subset are disconnected. While some versions of the network reliability problem are exactly solvable in polynomial time for the three-level networks that we study in this paper, our problem differs from the network reliability in that our goal is to construct an overlay network with considerations of both cost and quality-of-service as measured by end-to-end packet loss.

\subsection{Outline of the paper}

The remainder of this paper is organized as follows. 
In Section \ref{sec:problem}, we formally state the overlay network construction problem and show how the problem can be modeled as an Integer Program (IP). 
In Section \ref{sec:alg}, we describe our polynomial-time approximation algorithm {\tt Approx} that utilizes the technique of LP relaxation followed by rounding to obtain a near-optimal overlay network. In Section~\ref{sec:extensions}, we study various extensions to the overlay network construction problem to capture additional real-world constraints. 
In Section~\ref{sec:implement}, we show that our algorithm is capable of producing good overlay networks on variety of real-world trace data obtained from the Akamai live streaming network. Finally, in Section~\ref{sec:concl}, we conclude with directions for future research.   

 \section{The overlay network construction problem}
 \label{sec:problem}

In this section, we formally define the overlay network construction problem and show how the problem can be formulated as an integer program (IP). 
\subsection{Problem definition}
As an input to the problem, we are given the following.
\begin{itemize}
 \item A tripartite digraph $N=(V,E),$ where $V=S \cup R \cup D$ and $E = (S \times R) \cup (R \times D)$. The set $S$ denotes the set of sources, $R$ denotes the set of reflectors, and $D$ denotes the set of sinks (See Figure~\ref{tripartite}).  The nodes in the network represent the current deployment in the CDN of entry points (sources), reflectors, and edge servers  (sinks) that actually serve the streams to end-users. Each link $(i,j) \in E$ represents the underlying Internet path used for transmitting streams from node $i$ to node $j$. 

\item For each link (which corresponds to a path through the Internet), we are given the probabilities of packet loss on the links are denoted by
$$p_{ij} \in [0,1], \mathrm{for\ each\ }(i,j) \in E.$$ The loss probability reflects the odds that a packet sent on that link is lost or otherwise rendered useless (a packet that arrives significantly out-of-order or late is also useless). The link loss probabilities are typically measured by a software component residing at each node that sends a sequence of test packets to estimate the loss.

\item We are given a set of live streams where each stream has a specified source in $S$ and a specified subset of sinks in $D$ that require the stream. Though in reality multiple streams can originate at a physical entry point deployment or end at a physical edge server deployment, we assume that {\em exactly one stream originates at each source and exactly one stream ends at each sink\/}. We make this assumption without  any loss of generality since each source (resp., each sink)  can be replicated a sufficient number of times so that we have exactly one stream per source (resp., sink).  Once this modification is made, we let $n$ be the  maximum number of nodes in any level of the network, i.e., $n = \max (|D|, |R|)$ since $|S| \leq |D|$. 

\item Next, we are given the costs associated with routing the streams. Each stream incurs a cost for being transmitted over each link. The cost is represented by 
$$c^{k}_{ij} \in  {\R^{+}},$$ which is the cost of transmitting the stream originating at source $k \in S$ through link $(i,j) \in E$.
 Note that the costs of carrying different streams over a given link can be different depending on how they are encoded. The transmissions costs can also vary from link to link in accordance with the contracts between the CDN and the ISPs that provide the bandwidth. In addition to the transmission costs, we assume that there is a fixed cost for using a reflector to route one or more streams denoted by 
$$r_i \in {\R^{+}}, \mathrm{\ for\ each\ }i \in R.$$

\item The reflectors of the overlay network must obey capacity constraints that derive from hardware and software limitations.  First, we model the fan-out constraints on each reflector. For each reflector $i \in R$, reflector $i$ can simultaneously transmit to at most $F_i$ different sink nodes. In addition to fan-out constraints, one can also place an upper bound on the bandwidth (in bits per second) that can be transmitted through each reflector node. We extend our results to capture this additional constraint later in Section~\ref{sec:extensions}.

\item The goal of the overlay network is to simultaneously route each stream from its respective source to its subset of sinks with a minimum acceptable quality-of-service. The primary metric for quality of service (QoS) that we consider in this paper is end-to-end packet loss. To this end, we are given an {\em  end-to-end loss threshold\/} that represents the maximum acceptable end-to-end packet loss for each stream sent from its source $k$ to a sink $j$. The thresholds are represented by $$\Phi_{j}^k \in [0,1], \mathrm{for\ } k \in S\mathrm{\ and\ } j \in D.$$ Note that in this framework a given stream could require different levels of QoS at different sinks.
\end{itemize}
Given the digraph $N=(V,E)$, stream information, costs, capacities, link packet loss probabilities, and end-to-end loss thresholds as outlined above, the goal of overlay network construction is to create a set of paths $\Pi$ that can be used to simultaneously route each stream from its source to its subset of sinks such that capacity and end-to-end loss thresholds are met and the total cost is minimized.

%The problem is to find a minimum
%cost subnetwork such that when we send a packet, which is lost at
%each edge with some given probability, we are still assured that each
%sink will receive at least one copy of the packet with probability at least equal to the demand. The primary
%difference from previous network flow problems is that we don't have
%preservation of flow at each node. Instead if a flow is received at a
%reflector $i \in R$, it can be sent simultaneously to as many neighbors as its fan-out $F_i$. The cost of %routing along an arc may depend upon the commodity being sent. We
%describe an algorithm which approximates this min cost integer flow
%problem. This problem can model {\it SET COVER}. Thus the best
%solution in terms of cost that we can hope for, unless $\NP \subset
%\DTIME(n^{O(\log\log{n})})$, is a $O(\log{n})$
%approximation~\cite{Feige98}. Our problem is more general than set cover in several ways. We introduce 
%fan-out constraints on the reflectors (effectively, each set can cover only some of its elements). 
%We also have costs, both on the reflectors themselves and on covering a sink with a reflector, 
%and we require that each sink must be covered by multiple reflectors (typically single coverage 
%is not enough) which ensure at least the required success probability. We present an LP rounding %solution to the
%problem which has a guarantee of $O(\log{n})$ approximation on the
%cost and violates the probability and fan-out constraints by small
%constants.  

\subsection{Integer programming formulation}
\label{subsec:ipf}
We formulate the overlay network construction problem as an integer program (IP) as follows.  We
use $y_i^k$ as the indicator variable for the delivery of the $k$-th
stream to the $i$-th reflector, $z_i$ as the indicator variable for
utilizing reflector $i$ and $x_{ij}^k$ as the indicator variable for
delivering the $k$-th stream to the $j$-th sink through the $i$-th
reflector. Consider a stream that 
originates from source $k$ that goes through reflector $i$ to reach sink $j$.  As noted earlier, the loss experienced by the packets in the stream when transmitted over link $(k,i)$ (resp., link $(i,j)$) is $p_{ki}$ (resp., $p_{ij}$). Assuming that packet loss in the two links are independent, the loss experienced by the stream along the entire path from source $k$ to sink $j$  is $p_{ki}+p_{ij}-p_{ki}p_{ij}$. Since it is more convenient to work with the logarithms of probabilities, we
transform the packet loss probability along paths into weights where 
$w_{ij}^k= -
\log{(p_{ki}+p_{ij}-p_{ki}p_{ij})}$.  In other words 
$w_{ij}^k$ is the negative log of the probability of packet loss along path $(k,i,j)$.    
Likewise, we define $W_j^k = -\log{\Phi_j^k}$ to be the {\it weight threshold\/} of the stream originating at source $k$ and ending at sink $j$, where $\Phi_j^k$ is the corresponding end-to-end loss threshold\footnote{Note that both $w_{ij}^k$ and $W_j^k$ can take the value of $+\infty$ as defined. In practice, we use a sufficiently large finite value instead.}.
 Thus we are able to write
the IP:
\[
\begin{array}{lll}
\min \hspace{3mm} {\displaystyle \sum_{i \in R} r_i z_i + \sum_{i \in R} \sum_{k \in S} c_{ki}^k y_i^k 
+ \sum_{i \in R}\sum_{ k \in S} \sum_{ j \in D} c_{ij}^k x_{ij}^k} & &\vspace{-3mm}\\
s.t. & & \vspace{3mm} \\ \vspace{3mm}
&\hspace{-75mm} (1) &\hspace{-65mm} y_{i}^k \leq z_i \; \; \forall i \in R, \; \forall k \in S \\ \vspace{3mm}
&\hspace{-75mm} (2) & \hspace{-65mm}  x_{ij}^k \leq y_i^k \; \; \forall i \in R,\; \forall j \in D,\; \forall k \in S \\ \vspace{3mm}
&\hspace{-75mm} (3) &\hspace{-65mm} \sum_{k \in S} \sum_{j \in D} x_{ij}^k \leq F_i z_i \; \; \forall i \in R \\ \vspace{3mm}
&\hspace{-75mm} (4) &\hspace{-65mm}  \sum_{j \in D} x_{ij}^k \leq F_i y_i^k \; \; \forall i \in R,\; \forall k \in S \\ \vspace{3mm}
&\hspace{-75mm} (5) &\hspace{-65mm}  \sum_{i \in R} x_{ij}^k w_{ij}^k \geq W_{j}^k  \; \; \forall j \in D,\; \forall k \in S \\ \vspace{3mm}
&\hspace{-75mm} (6) &\hspace{-65mm} x_{ij}^k \in \{0,1\}, \; y_i^k \in \{0,1\}, \; z_i \in \{0,1\}
\end{array}
\]

Constraints (1) and (2) force us to pay for the reflectors we are
using, and to transmit packets only through reflectors that are in
use. Constraint (3) encodes the fan-out restriction. Constraint (4) is
redundant in the IP formulation, but provides a useful cutting plane
in the LP rounding algorithm that we present in Section~\ref{sec:alg}. 
Constraint (5) are the ``weight constraints'' that capture the end-to-end loss requirements for QoS as shown in Claim~\ref{clm:threshold} below. Note that since we have replicated the sinks such that exactly one stream ends at each sink, there is exactly one weight constraint per sink. Thus, there are a total of $|D| \leq n$ weight constraints, where $n = \max (|R|, |D|)$. Constraint (6) is the integrality constraint for the variables. The set of paths $\Pi$ that is the output of overlay network construction can be extracted from the solution to the IP above by routing a stream from its source $k$ through reflector $i \in R$ to sink $j \in D$ if and only if $x_{ij}^k$ equals $1$ in the IP solution. The {\em cost objective function\/} that is minimized represents the total cost of operating the overlay network and is the sum of three parts: the fixed cost of using the reflectors, the cost of sending streams from the sources to the reflectors, and the cost of sending streams from the reflectors to the sinks.

\begin{claim}
\label{clm:threshold}
If constraint (5) holds, then for each stream originating at source $k$ and destined for sink $j$,  the end-to-end packet loss is at most the corresponding end-to-end  loss threshold $\Phi_j^k$. 
\end{claim}

\begin{proof}
Since packet loss on different links are assumed to be independent,  the end-to-end loss probability of the reconstructed stream from source $k$ to sink $j$ is the product of the loss probabilities along each of its edge-disjoint paths in $\Pi$. Since we use the negative logarithm of the loss probabilities as weights, the logarithm of the product of loss probabilities is equal to the sum of the corresponding weights. Specifically,  the LHS of constraint (5) is the sum of the weights of the edge-disjoint paths in $\Pi$ from source $k$ to sink $j$, which equals the  negative logarithm of the end-to-end loss probability of the stream from $k$ to $j$. Note that $W^{k}_{j}$  represents the negative logarithm of the acceptable end-to-end loss threshold $\Phi_j^k$. Therefore, asserting that the LHS value is at least $W^{k}_{j}$ captures the end-to-end threshold requirement for the stream.
\end{proof}

\begin{claim}{\it
In the IP formulation constraints (1),(2),(3) and (6) dominate (4). }
\end{claim}

\begin{proof}
We look at cases for $z_i$. 
\begin{itemize}
\item[1)] If $z_i=0$, then from (1) and (6) we get $y_i^k=0$ for $\forall k \in S$. Now
from (2) and (6) we get $x_{ij}^k=0 \; \;\forall k \in S$ and $\forall j \in D$. Thus (4) is implied.
\item[2)] If $z_i=1$, then if $y_i^k=0$ we still have $x_{ij}^k=0 \; \; \forall j \in D$,
which means
\[ \sum_{j \in D} x_{ij}^k = 0 \]
If $y_i^k=1$ then from (3) we have
\[   \sum_{k \in S} \sum_{j \in D} x_{ij}^k \leq F_i \]
which means that $\forall k \in S$
\[ \sum_{j \in D} x_{ij}^k \leq F_i \]
\end{itemize}
Which concludes the proof.
\end{proof}

\section{An approximation algorithm for overlay network construction}
 \label{sec:alg}
 
 To motivate the need for an approximation algorithm, we first show a hardness result for computing the optimal solution for the overlay network construction problem. 
  \begin{theorem}
\label{thm:lowerbound}
The overlay network construction problem is NP-hard. Further, there is no polynomial time algorithm that can achieve cost that is within a $(1 - o(1)) \ln n$ factor of optimal, unless\footnote{Note that $NP  \subset DTIME(n^{O(\log\log{n})})$ is a weaker requirement than $P = NP$, thus yielding a stronger result.}  $NP  \subset DTIME(n^{O(\log\log{n})})$.
\end{theorem}

\begin{proof}
The set cover problem that is known to be NP-hard \cite{GJ79} is as follows. Given a set of elements and a collection of sets containing those elements, the goal of set cover is to select the smallest number of sets such that every element is included in at least one of the selected sets.  The set cover problem is a special case of the overlay network construction problem as shown below. Let each set correspond to a reflector and each element correspond to a sink. Further, let there be one source (labeled 1) that originates a single stream that must be sent to all sinks with $W^1_j = 1$, for all $j \in D$. Next, let $w^1_{ij} = 1$ if the set represented by reflector $i$ contains the element represented by sink $j$, and zero otherwise. (Note that  $w^1_{ij}$ is set to $1$ by making the probability of loss on path $(1, i, j)$ to be $1/2$ and $w^1_{ij}$ is set to $0$ by making the probability of loss on path $(1, i, j)$ to be $1$.)  Finally, let $r_i = 1$, for all $i \in R$, let all other costs be zero, and let fan-out $F_i$ be unbounded for all $i \in R$.  It is easy to see that solving the above special case of the overlay network construction problem with the smallest cost is equivalent to finding the smallest collection of reflectors in $R$ that cover all the elements in $S$. Thus, overlay network construction is NP-hard. Further, it is known that the there is no polynomial time algorithm that can approximate the set cover problem within a $(1 - o(1)) \ln n$ factor of optimal, unless  $NP  \subset
DTIME(n^{O(\log\log{n})})$ \cite{LundY94,Feige98}. Thus, the same inapproximability result also holds for the overlay network construction problem.
\end{proof}
With Theorem~\ref{thm:lowerbound} in mind, we now describe a polynomial-time approximation algorithm {\tt Approx} that produces a solution that has cost within a $O(\log n)$ factor of optimal while satisfying the fan-out and weight constraints within constant factors, i.e., the algorithm has the best possible approximation ratio to within constant factors. Our approximation algorithm {\tt Approx} works in two phases. In the first phase, the integer program (IP)  described in Section~\ref{subsec:ipf} is ``relaxed'' to obtain a linear program (LP). Specifically, the  LP relaxation is obtained by substituting 
the integrality constraints (6) in the IP with 
\[
x_{ij}^k \in [0,1], \; y_i^k \in [0,1], \; z_i \in [0,1]
\]
That is, the above variables can now take fractional values, rather than just $0$ or $1$. We solve the LP optimally and find a fractional solution denoted by
\[ (\hat{z}_i, \; \; \hat{y}_i^k,  \; \; \hat{x}_{ij}^k) \] 

In the second phase, we find a solution to the IP by ``rounding'' the fractional LP solution to integral values using a two-step rounding process: a randomized rounding step (Section~\ref{sec:randround}) followed by a modified version of a Generalized Assignment Problem (GAP) approximation (Section~\ref{sec:gap}). Once the rounding process is complete, we establish that the integral solution obtained through rounding is a provably good approximate solution for the original IP, which in turn provides a provably good overlay network for simultaneously routing all the streams from their sources to their respective destinations. 

\subsection{Randomized rounding}
\label{sec:randround}

We apply the following randomized rounding procedure to obtain the values $(\bar{z}_i, \; \; \bar{y}_i^k,  \; \; \bar{x}_{ij}^k)$.
We use parameter $c>1$, which will be determined later, as a preset
multiplier.
\begin{itemize}
\item[(1)] {\em Compute $\dot{z}_i$:} $\forall i \in R$, set $ \; \dot{z}_i = \min (\hat{z}_i c\log{n}, 1)$ 
\item[(2)] {\em Compute $\dot{y}_i^k$:} $\forall i \in R, \; \forall k \in S$, if  $\dot{z}_i = 0$ then  set $\dot{y}_i^k = 0$, else set
$$\dot{y}_i^k = \min \left(\frac{\hat{y}_i^k c\log{n}}{\dot{z}_i}, 1\right)$$
\item[(3)] {\em Compute $\bar{z}_i$:} We round $\bar{z}_i=1$ with probability $\dot{z}_i$ and 0 otherwise.
\item[(4)] {\em Compute $\bar{y}_i^k$:} If $\bar{z}_i=1$ then round $\bar{y}_i^k=1$ with probability $\dot{y}_i^k$ and 0 otherwise.
\item[(5)] {\em Compute $\bar{x}_{ij}^k$:} If $\dot{z}_i=\dot{y}_i^k=1$ set $\bar{x}_{ij}^k=\hat{x}_{ij}^k$ \\
            else if $\bar{y}_i^k=1$ set $\bar{x}_{ij}^k=\frac{1}{c\log n}$ 
with probability $\hat{x}_{ij}^k / \hat{y}_i^k$ and 0 otherwise.
\item[(6)] Set all the variables $\bar{z}_i$, $\bar{y}_i^k$, and $\bar{x}_{ij}^k$ not set in the above steps to 0.
\end{itemize}
The only fractional values left after this procedure 
are $\bar{x}_{ij}^k$. As outlined later in Section~\ref{sec:gap}, to round the $\bar{x}_{ij}^k$'s we will apply a modified version of the 
Generalized Assignment Problem (GAP) approximation due to Shmoys and Tardos~\cite{ShmoysT93}.
The GAP rounding will preserve the cost and violate the fan-out and weight constraints by
at most an additional constant factor.  

\subsubsection{Analysis of the randomized rounding} 
\label{sec:rounding}
We bound the expected cost of the solution after randomized rounding  in terms of the optimal cost as follows.
Let $\hat{C}$ denote the value of the cost objective function for our fractional solution $(\hat{z}_i, \; \; \hat{y}_i^k,  \; \; \hat{x}_{ij}^k)$ obtained by solving the LP relaxation. Likewise, let
$\bar{C}$ be the value of the cost objective function after the randomized rounding procedure, i.e., $\bar{C}$ is the value obtained by evaluating the objective function using values $(\bar{z}_i, \; \; \bar{y}_i^k,  \; \; \bar{x}_{ij}^k)$. 
Finally, let $C^{OPT}$ the optimal objective function value obtained by solving the IP.
From steps (1) and (3) of the rounding procedure we conclude that
\begin{equation}
\E[\bar{z}_i]  = \dot{z}_i \leq \hat{z}_i c\log{n}. \label{eq:zbar}
\end{equation}
From step (4) of the rounding procedure, we conclude that
\begin{equation}
\E[\bar{y}^k_i] = \dot{z}_i \cdot \min \left(\frac{\hat{y}_i^k c\log{n}}{\dot{z}_i}, 1\right) \leq \hat{y}^k_i c\log{n}.  \label{eq:ybar}
\end{equation}
Further,  in all cases as shown below,
\begin{equation}
\E(\bar{x}_{ij}^k) = \hat{x}_{ij}^k.
\label{eq:xbar}
\end{equation}
 In the case where $\dot{z}_i=\dot{y}_i^k=1$, we deterministically set $\bar{x}_{ij}^k=\hat{x}_{ij}^k$ and hence $\E(\bar{x}_{ij}^k) = \hat{x}_{ij}^k$. Else, we have the following two cases to consider.
\begin{enumerate}
\item If $\dot{y}_i^k < 1$, it follows that $$\dot{y}_i^k =  \frac{\hat{y}_i^k c\log{n}}{\dot{z}_i}.$$
\item If $\dot{z}_i< 1$, then $\dot{z}_i = \hat{z}_i c \log n$. Thus,  since $\hat{y}_i^k  \leq \hat{z}_i$ due constraint (1)  of the LP, 
\[ \frac{\hat{y}_i^k c\log{n}}{\dot{z}_i}
= \frac{\hat{y}_i^k c\log{n}}{\hat{z}_i c \log n}  \leq  1.\]  
That is, it is again true that $$\dot{y}_i^k =  \min \left(\frac{\hat{y}_i^k c\log{n}}{\dot{z}_i}, 1\right) = \frac{\hat{y}_i^k c\log{n}}{\dot{z}_i}.$$
\end{enumerate}
 Therefore, in both of the above cases,
\[
\E(\bar{x}_{ij}^k) = \dot{z}_i \cdot \dot{y}_i^k  \cdot (\hat{x}_{ij}^k / \hat{y}_i^k) \cdot \frac{1}{c\log n} = \dot{z}_i \cdot \frac{\hat{y}_i^k c\log{n}}{\dot{z}_i}  \cdot (\hat{x}_{ij}^k / \hat{y}_i^k) \cdot \frac{1}{c\log n} =  \hat{x}_{ij}^k . 
\]
Using the linearity of expectations, the Equations~\ref{eq:zbar},\ref{eq:ybar}, and \ref{eq:xbar} above imply
\[ \E[\bar{C}] \leq c\log{n} \cdot \hat{C} \leq c\log{n} \cdot C^{OPT}. \]
Thus we have the following lemma.

\begin{lemma}{\it The expected cost after the randomized rounding step is  $O(\log{n})$ times the optimal cost. }
\label{lem:costbnd}
\end{lemma}
Now we will show that with high probability all weight constraints are met to within a small constant factor. %Combining this with the GAP approximation of Section~\ref{sec:gap} will yield a
%solution to the IP that has a cost at most $c\log{n}$ times optimal and violates the
%fan-out and weight constraints by at most a factor of 4.
By high probability, we mean a probability of more than $1 - 1/n$, where $n = \max (|R|, |D|)$ as defined earlier. Recall that the constraints (5) are the ``weight constraints'' as shown below:
\[ \sum_{i \in R} x_{ij}^k w_{ij}^k \geq W_{j}^k  \; \; \forall j \in D,\; \forall k \in S. \]
Let random variable  $\overline{W}_{j}^k$ be the  LHS of the weight constraint evaluated at $\bar{x}_{ij}^k$, i.e., 
\[\overline{W}_{j}^k \stackrel{\Delta}{=}  \sum_{i \in R} \bar{x}_{ij}^k w_{ij}^k.\]  We show that after the randomized rounding step the following holds with high probability: $\overline{W}_{j}^k  \geq (1 - \delta) W_j^k$, for all $j$ and $k$ and a small constant $\delta > 0$.
To provide a probabilistic bound for $\overline{W}_{j}^k$, we use a version of the {\em Hoeffding-Chernoff bound} \cite{Hoeffding63,MotwaniR95} below.

\begin{theorem}[Hoeffding-Chernoff bound] 
\label{theorem:4.2}
Let $v_i$ be a set of independent random variables where for all $i$ either $v_i \in [0,1]$ or $Variance(v_i) = 0$. Let $0 < \delta < 1$, let $S=\sum_i v_i$ and 
$\mu = \E \left[ \sum_i v_i \right]$ then
\[
\begin{array}{l} 
\Pr(S < (1-\delta) \mu) \leq exp\left(-\frac{\delta^2\mu}{2}\right) \vspace{3mm}\\
\Pr(S > (1+\delta) \mu) \leq exp\left(-\frac{\delta^2\mu}{3}\right)
\end{array}
\] 
\end{theorem}

\begin{proof}
This theorem is the standard Hoeffding-Chernoff bound with a small modification. The standard bound requires all random variables $v_i \in [0,1]$. However, we allow random variables with zero variances, i.e., those that take a single value with probability 1, to be in any range. Clearly, any such $v_i$ with zero variance can decomposed into $\lceil v_{i}  \rceil$ variables that are each deterministically set to $\frac{v_{i}}{\lceil v_{i} \rceil} \in \left[0, 1 \right]$ with probability 1. Applying the standard bound after this decomposition yields our theorem.
\end{proof}

We define random variable $v_{ij}^k \stackrel{\Delta}{=} (c \log n) \bar{x}_{ij}^k \frac{w_{ij}^k}{W_j^k}$. In order that we may use Theorem~\ref{theorem:4.2}, we need to establish the following two criteria.
\begin{enumerate}
\item The random choices made in computing $v_{ij}^k$  and $v_{i'j}^k$, for $i \neq i'$, are independent and not shared. Each random variable $v_{ij}^k$ is computed using random variable $\bar{x}_{ij}^k$ which in turn depends on the random choices made in computing $\bar{y}_{i}^k$.  But, for any $i  \neq  i'$, random variable $v_{ij}^k$  is independent of $v_{i'j}^k$, since the random choices in computing $\bar{y}_{i}^k$ and $\bar{z}_i$ are distinct from those made in computing $\bar{y}_{i'}^k$ and $\bar{z}_{i'}$.  
 
\item  Without loss of generality we can assume $w_{ij}^k \leq W_j^k$ since it never helps to have more weight on
a source-to-sink path than the weight that the sink itself demands. When  $\bar{x}_{ij}^k$ is probabilistically set, it is set to either $\frac{1}{c\log n}$ or $0$. It follows that $v_{ij}^k =(c \log n) \bar{x}_{ij}^k \frac{w_{ij}^k}{W_j^k} \in \left[0, 1\right]$.  Otherwise, $\bar{x}_{ij}^k$ is deterministically set to $\hat{x}_{ij}^k$ with probability $1$. The range of $v_i$ can be arbitrary in this case, since the variance of $v_i$ is zero.
\end{enumerate}
Using Equation~\ref{eq:xbar}, the expected value of $v_{ij}^k$ is
\[ \E[v_{ij}^k]= c\log{n} \cdot \frac{w_{ij}^k}{W_j^k} \cdot \hat{x}_{ij}^k.  \]
Since the weight constraints are satisfied by the LP solution,
\begin{equation}
\sum_{i \in R} \frac{w_{ij}^k}{W_j^k} \cdot \hat{x}_{ij}^k \geq 1.
\label{eq:wtineq} 
\end{equation}
Thus, 
\begin{equation}
 \mu \stackrel{\Delta}{=}  \E\left[ \sum_{i \in R} v_{ij}^k \right] = \sum_{i \in R} \E[v_{ij}^k] = c\log{n} \cdot 
\sum_{i \in R} \frac{w_{ij}^k}{W_j^k} \cdot \hat{x}_{ij}^k \geq c\log{n}. 
\label{eq:meanbnd}
\end{equation} 
Using the Hoeffding-Chernoff bound of Theorem~\ref{theorem:4.2}, we get the following chain of inequalities: 
\[
\begin{array}{l} 
\Pr(\overline{W}_j^k < (1-\delta)W_j^k) = \Pr\left(\sum_{i \in R} w_{ij}^k \cdot \bar{x}_{ij}^k < (1 - \delta) W_j^k\right)  \\
\leq \Pr\left(\frac{\sum_{i \in R} w_{ij}^k \cdot \bar{x}_{ij}^k}{W_j^k} < (1 - \delta) \right), \mathrm{by\ dividing\ both\ sides\ by\ W_j^k}\\
 \leq \Pr\left(\sum_{i \in R} \frac{w_{ij}^k \cdot \bar{x}_{ij}^k}{W_j^k}< 
(1-\delta)\sum_{i \in R} \frac{w_{ij}^k\cdot \hat{x}_{ij}^k}{W_j^k} \right), \mathrm{using\ Equation~\ref{eq:wtineq}}\\
=\Pr(\sum_{i \in R} v_{ij}^k  \leq (1-\delta) \mu) \leq exp\left(-\frac{\delta^2\mu}{2}\right). 
\end{array}
\]
Thus, for a fixed $j$ and $k$ we bound the probability that  the corresponding weight constraint is not met to within a factor of $1 - \delta$ using Equation~\ref{eq:meanbnd} and Theorem~\ref{theorem:4.2} as follows:
\[ \Pr\left(\overline{W}_{j}^k  < (1 - \delta) W_j^k\right) \leq
e^{\left(-\frac{\delta^2 \cdot c\log{n} }{2}\right)}=\frac{1}{n^{\delta^2 \cdot c /2}}. \]
Since there are at most $n$ weight
constraints, i.e., one constraint per sink, we set $\delta^2 \cdot c = 4$. Specifically, if we set $\delta=1/4$ and
$c \geq 64$, all weight constraints are satisfied to within a factor of $1 - \delta$ with probability at least $1 - \frac{n}{n^{\delta^2 \cdot c /2}} \geq 1 - \frac{n}{n^2} \geq 1 - \frac{1}{n} $. Thus, we can state the following lemma.
\begin{lemma}
\label{lem:weightbnd}
{\it
If we set $c \geq 64$ then after the randomized rounding procedure, all weight constraints are met to within a factor of $\frac{3}{4}$, with high probability. } That is,  
$\overline{W}_{j}^k  \geq \frac{3}{4} W_j^k$,  for all $j$ and $k$, with high probability.
\end{lemma}       

Finally, we show that all fan-out constraints are met to within a factor of two. Specifically, our goal is to show that, with high probability, after the randomized rounding step the following holds:
\[ \sum_{k \in S} \sum_{j \in D} \bar{x}_{ij}^k \leq2  F_i  \; \; \forall i \in R. \]
We want to again apply the Hoeffding-Chernoff bound. Unfortunately, the random variables $\bar{x}_{ij}^k$ in the above sum are not all independent. Specifically, $\bar{x}_{ij}^k$ and $\bar{x}_{ij'}^k$, for $j \neq j'$, are not independent as both depend  on the random choices made in computing $\bar{y}^k_i$. For instance,  if $\bar{y}^k_i=1$ there is a higher probability 
 for all $j \in D$ that $\bar{x}_{ij}^k$ will be rounded to $1/c \log{n}$. However, $\bar{x}_{ij}^k$ are 
obtained by a two stage process in which first $\bar{y}_i^k$ is rounded to $0$ or $1$ and
then $\bar{x}_{ij}^k$ is rounded if and only if $\bar{y}_i^k=1$. 
We will use two claims to prove the next lemma. Let random variable $\E\left[\sum_{k \in S}\sum_{j \in D} \bar{x}_{ij}^k|\bar{y}_i^k\right] $ be defined over the probability space of  $\bar{y}_i^k$, i.e.,  the random variable $\E\left[\sum_{k \in S}\sum_{j \in D} \bar{x}_{ij}^k|\bar{y}_i^k\right]$ is a function that maps $\bar{y}_i^k$ to the value of the conditional expectation given the values of $\bar{y}_i^k$.
\begin{claim}{\it
\label{claim:4.5}
For any reflector $i$,  we have
\[
\Pr\left(\E\left[\sum_{k \in S}\sum_{j \in D} \bar{x}_{ij}^k|\bar{y}_i^k\right] >
\frac{3}{2}F_i \right) \leq \frac{1}{2n^2}.
\] }
\end{claim}

\begin{proof}
In order to apply Theorem~\ref{theorem:4.2}, we set 
$$v^k_i \stackrel{\Delta}{=} \frac{c\log{n}}{F_i} \cdot \E\left[\sum_{j \in D} \bar{x}_{ij}^k|\bar{y}_i^k\right] .$$
Excluding the situation where $\dot{z}_i = \dot{y}_{i}^k  = 1$ where $\bar{y}^k_i= 1$ and $\bar{x}_{ij}^k$ is  set to $\hat{x}_{ij}^k$ with probability $1$ yielding a $v^k_i$ with zero variance, we have the following two cases.
 Either
$\bar{y}_i^k=0$ then
\[ \E\left[\sum_{j \in D} \bar{x}_{ij}^k|\bar{y}_i^k\right] =0 \]
Or, if $\bar{y}_i^k=1$ then from the cutting plane constraint (4) from the IP formulation we have
\[
\E\left[\sum_{j \in D} \bar{x}_{ij}^k|\bar{y}_i^k\right] = \sum_{j \in D}
\frac{1}{c\log{n}} \cdot \frac{\hat{x}_{ij}^k}{\hat{y}^k_i} \leq
\frac{F_i}{c\log{n}}
\]
It follows that in both these cases, 
$$v^k_i =  \frac{c\log{n}}{F_i} \cdot \E\left[\sum_{j \in D} \bar{x}_{ij}^k|\bar{y}_i^k\right] \in \left[0,1\right].$$
We know from Equation~\ref{eq:xbar} and the fan-out constraint (3) that
\[  \E\left[\sum_{k \in S}\sum_{j \in D} \bar{x}_{ij}^k\right]  = \sum_{k \in S}\sum_{j \in D} \E\left[\bar{x}_{ij}^k\right] = \sum_{k \in S}\sum_{j \in D} \hat{x}_{ij}^k \leq F_i. \]
Therefore, using the above equation and the linearity of expectation, we have
\begin{eqnarray}
 \mu & =  & \E\left[\sum_{k \in S} v^k_i \right] =  \frac{c\log{n}}{F_i}  \cdot \E\left[\sum_{k \in S}\E\left[\sum_{j \in D} \bar{x}_{ij}^k|\bar{y}_i^k\right]\right] \nonumber \\ &=&  \frac{c\log{n}}{F_i}  \cdot \E\left[\sum_{k \in S}\sum_{j \in D} \bar{x}_{ij}^k\right]  \leq \frac{c\log{n}}{F_i} \cdot  F_i  = c\log{n}.\label{eq:muupbnd}
 \end{eqnarray}
Now we use the Hoeffding-Chernoff bound of Theorem~\ref{theorem:4.2} and we get
\[
\begin{array}{l}
\Pr\left(\E\left[\sum_{k \in S}\sum_{j \in D} \bar{x}_{ij}^k|\bar{y}_i^k\right] >
\frac{3}{2}F_i\right) \\
= \Pr\left(\sum_{k \in S} v^k_i > \frac{3}{2} c \log n\right), \mathrm{multiplying\ both\ sides\ by\ \frac{c \log n}{F_i}} \\ 
\leq \Pr\left(\sum_{k \in S} v^k_i > \left(1+\frac{c \log n}{2 \mu}\right)\mu\right), \mathrm{using\ Equation~\ref{eq:muupbnd}}\\
\leq exp\left(\frac{-(c \log n)^2 \mu}{(2 \mu)^2 3}\right) \leq 
exp\left(\frac{-1}{12} c \log n\right)= n^{-c/12}. 
\end{array}
\]
By setting $c > 24$ we get
\[
\Pr\left(\E\left[\sum_{k \in S}\sum_{j \in D} \bar{x}_{ij}^k|\bar{y}_i^k\right] >
\frac{3}{2}F_i\right) < \frac{1}{2n^2}.
\]
Which concludes the proof of this claim.
\end{proof}

\begin{claim} {\it
\label{claim:4.6}
For some reflector $i$, suppose that the  $\bar{y}_i^k$ values are fixed such that the following holds: 
\[
\E\left[\sum_{k \in S}\sum_{j \in D} \bar{x}_{ij}^k|\bar{y}_i^k\right] \leq
\frac{3}{2}F_i
\]
Then for $c \geq 36$
\[
\Pr\left(\sum_{k \in S}\sum_{j \in D} \bar{x}_{ij}^k > 2F_i\right) \leq
\frac{1}{2n^2}
\]
}
\end{claim}
\begin{proof}
For a given reflector $i$, when all $\bar{y}_i^k$ are fixed then $\bar{x}_{ij}^k$ are independent. 
We now define 
$$v^k_{ij} \stackrel{\Delta}{=} 
\frac{c\log{n}}{F_i} \cdot \bar{x}_{ij}^k.$$ 
As in Claim~\ref{claim:4.5}, excluding the case where $v^k_i$ has zero variance,  we have $v^k_{ij}  \in [0,1],$ since $\bar{x}_{ij}^k \in [0, \frac{1}{c \log n}]$, $F_i \geq 1$,  and hence
\[v^k_{ij} =
\frac{c\log{n}}{F_i} \cdot  \bar{x}_{ij}^k  \leq \frac{c\log{n}}{F_i} \cdot  \frac{1}{c \log n} \leq 1. \]
Furthermore from the first part of this claim
\begin{equation}
 \mu \stackrel{\Delta}{=} \E\left[ \sum_{k \in S} \sum_{j  \in D} v^k_{ij} \right] \leq  \frac{3}{2}F_i \cdot \frac{c\log{n}}{F_i}
=  \frac{3c\log{n}}{2}.
\label{eq:muupbnd2}
\end{equation}
Thus we apply Hoeffding-Chernoff bound of Theorem~\ref{theorem:4.2} and we get
\[ 
\begin{array}{l}
\Pr\left(\sum_{k \in S}\sum_{j \in D} \bar{x}_{ij}^k > 2F_i\right) = \\
= \Pr\left(\sum_{k \in S} \sum_{j \in D} v^k_{ij} > 2 c \log n\right), \mathrm{multiplying\ both\ sides\ by\ \frac{c \log n}{F_i}} \\ 
\leq \Pr\left(\sum_{k \in S} \sum_{j \in D} v^k_{ij} > \left(1+\frac{c \log n}{2 \mu}\right)\mu\right), \mathrm{using\ Equation~\ref{eq:muupbnd2}}\\
\leq  exp(-(\frac{c \log n}{2 \mu})^2 \mu / 3) \leq exp(-c \log n/18), \mathrm{using\ Theorem~\ref{theorem:4.2}\ and\ Equation~\ref{eq:muupbnd2}}\\
\leq  n^{\frac{-c}{18}}.
\end{array}
\]
Which by setting $c \geq 36$ concludes the proof of this claim.
\end{proof}

 For a given reflector $i$, we apply the union bound to the bounds in Claims~\ref{claim:4.5} and ~\ref{claim:4.6} to conclude that the 
 \[ \Pr\left( \sum_{k \in S} \sum_{j \in D} \bar{x}_{ij}^k > 2 F_i\right) \leq \frac{1}{2n^2} + \frac{1}{2n^2} = \frac{1}{n^2}.
 \]
 Since there are at most $n$ reflectors, the probability that some reflector exceeds the fan-out constraints by more than an factor of $2$ is at most $n \cdot \frac{1}{n^2} = \frac{1}{n}$. Thus. we can state the following lemma.
\begin{lemma}
\label{lem:fan-outbnd}
{\it
If we set $c \geq 36$ then after the randomized rounding procedure 
all fan-out constraint are met to within a factor of 2, with high probability. }
\end{lemma}

%\begin{figure}[htbp]
%        \centerline{\psfig{figure=GAPpic.eps,height=2.2in,width=3.2in}}
%        \caption{$\bar{x}_{ij}^k$ fractional solution conversion network}
%\label{}
%\end{figure}

\subsection{Rounding by modified GAP
approximation}
\label{sec:gap}
As the last step in the rounding process, we describe how to
convert the $\bar{x}_{ij}^k$  produced by the randomized rounding step to an integral solution. This solution will
violate the fan-out constraints by an additional constant factor, so that all weight constraints will be met to within a combined factor of $1/4$ with high probability.
As before  let $\bar{C}$ denote the cost value achieved by solution $(\bar{z}_i, \; \; \bar{y}_i^k,  \; \; \bar{x}_{ij}^k)$ after the randomized rounding step.
\begin{figure}[htbp]
        \centerline{\psfig{figure=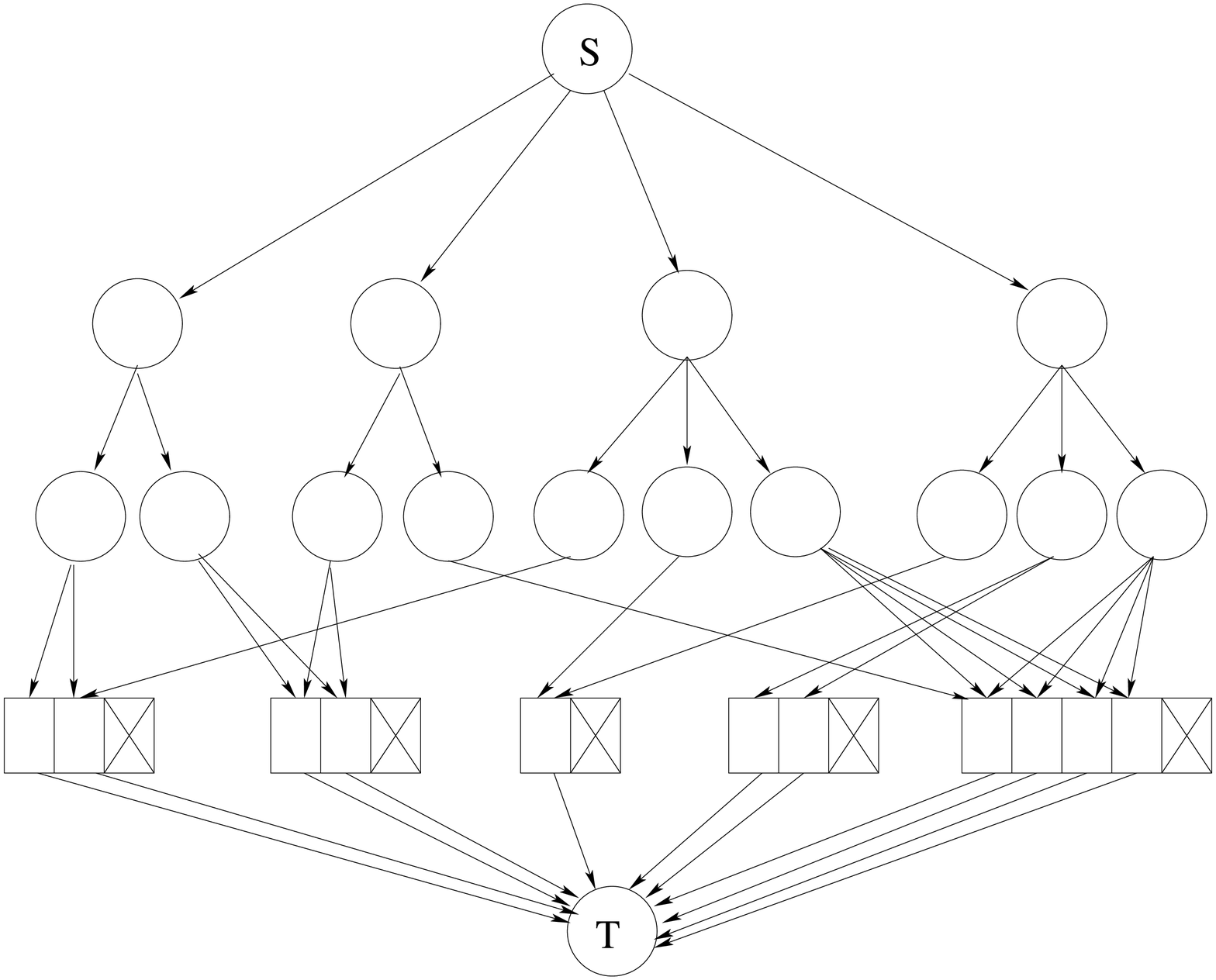,height=2.2in,width=3.2in}}
        \caption{Converting $\bar{x}_{ij}^k$ to an integral solution using a GAP flow graph}
\label{fig:conversion}
\end{figure}
Using the values $\bar{x}_{ij}^k$, we create a 
five-level ``GAP flow graph'' with edge capacities and edge costs that will help us perform the rounding (See Figure~\ref{fig:conversion}). The flow graph consists of a single node labeled $\bf{S}$ in the first level and a single node labeled $\bf{T}$ in the fifth level.  At the second level, we create a vertex for each reflector in $R$. The node $\bf{S}$ is connected to  each reflector $i$ in the second level with an edge of capacity equal to $2 F_i$ and zero cost, where $F_i$ is the maximum fan-out. The third level consists of nodes representing all (reflector i, sink j) pairs  
such that $\bar{x}_{ij}^k \neq 0$. We add an edge with capacity $1$ and zero cost from a node in the second level representing reflector $i$ to all nodes in the third level $(i,j)$ such that  $\bar{x}_{ij}^k \neq 0$.
In the fourth level we represent each sink $j$ as a collection of boxes where the number of boxes is equal to
\[ s_j = \left\lceil 2 \sum_{i \in R}  \bar{x}_{ij}^k \right\rceil. \]
Note that since each sink receives only one stream fixing $j$ automatically fixes $k$. We order the $w_{ij}^k$ for each sink in non-increasing order.  That is, WLOG we assume that
\[ w_{1j}^k \geq w_{2j}^k \geq \ldots \]
This ordering of weights induces a corresponding ordering on the nonzero $\bar{x}_{ij}^k$ values.
With each box we associate an interval of weights in this ordering. The corresponding $\bar{x}_{ij}^k$ values are also similarly associated with that box.  In associating weights and values, we ensure that the $\bar{x}_{ij}^k$ values associated with each box sum to exactly $1/2$, except possibly the last box. We associate weights and values with each box using the following process.
Let $t$ be the first index for which
\[ \sum_{i=1}^t \bar{x}_{ij}^k > \frac{1}{2}.\]
 We associate with the first box the weight interval $[w_{1j}^k,w_{tj}^k]$ and the corresponding portion of $\bar{x}_{ij}^k$ values that add up to exactly $1/2$. Next, we compute an index $r \geq t$ and associate an interval  with the second box as follows. Let $x'= \sum_{i=1}^t \bar{x}_{ij}^k -1/2$.
If $x'>1/2$ we set $r=t$ and associate with the second box the weight interval $[w_{tj}^k,w_{tj}^k]$ and the portion of $\bar{x}_{tj}^k$ of value $1/2$ . Otherwise we 
set $r$ to be the smallest index such that
\[ x'+\sum_{i=t+1}^r \bar{x}_{ij}^k > \frac{1}{2}.\] And, we associate the second box the weight interval $[w_{tj}^k,w_{rj}^k]$ and the corresponding portion of the $\bar{x}_{ij}^k$ values, $t \leq i \leq r$, that total to $1/2$. 
We continue with this process until we
associate all the boxes with weight intervals, except possibly the last box that may not be associated with a weight interval. Note the above algorithm ensures that the total of the $\bar{x}_{ij}^k$ values associated with each box is exactly $1/2$, with the possible exception of the last box.  We then eliminate the last box  for each sink.  Using Lemma~\ref{lem:weightbnd} and our assumption that $w_{ij}^k \leq W_j^k$, we can conclude that for all sinks $j$, with high probability, $$\sum_{i\in R} \bar{x}_{ij}^k \geq \sum_{i \in R} \frac{w_{ij}^k}{W_j^k} \cdot \bar{x}_{ij}^k \geq 3/4.$$ Thus, since the $\bar{x}_{ij}^k$'s add up to more than $1/2$ for each sink, there is at least one box that remains for each sink after eliminating the last box.
Using the weight interval assignments, we connect each node in the third level representing a (reflector i, sink j) pair
to the corresponding set of boxes that represent sink $j$ on the fourth level. Let $b$ be a box in the set of boxes that represents sink $j$ that is assigned the weight interval $[w_{tj}^k,w_{rj}^k]$ for some $t \leq r$. For each $t \leq i \leq r$, we place an edge with  capacity
1/2 and cost $c^k_{ij}$ between the third-level node representing pair $(i,j)$ and the fourth-level node representing box $b$. Finally, we connect all the (non-eliminated) boxes to node $\bf{T}$ with edges of
capacity 1/2  and zero cost. 

Using our construction, a maximum flow can be routed from node $\bf{S}$ to node $\bf{T}$ as follows. We start by routing a flow of $\frac{1}{2}$ from each each box in the fourth level to $\bf{T}$. We extend these flows to the third level by using $\bar{x}_{ij}^k$ values associated with each box as the flow values. This flow can be further extended to the second and first levels of the graph in the obvious way, following flow conservation at each node and the capacity constraints. From Lemma~\ref{lem:fan-outbnd}, we know that the capacity of the $2 F_i$ on the edges from first to the second level of the graph are not violated, with high probability.  Note that this flow saturates the cut of edges that come into $\bf{T}$ since each of these edges have a flow of $1/2$, hence the flow is maximum. The maximum flow routed in this fashion may have fractional flow values on the edges from the third to the fourth level and these values correspond to the $\bar{x}_{ij}^k$ values. This flow has a total cost of at most $\bar{C}$, since eliminating boxes can only reduce the total flow and hence the total cost. However, since all edge capacities are either integral or $\frac{1}{2}$,  there
exists a minimum cost maximum flow with flow variables that equal only 0, 1/2 or 1 that has a total cost that is at most the cost of the original max flow. Thus, the new min-cost maximum flow has a total cost of 
at most $\bar{C}$.  We find such a maximum flow with minimum cost using a known polynomial time algorithm \cite{AhujaMO93} and that provides us new flow values that we define to be $\tilde{x}_{ij}^k$ that equal either 0, 1/2, or 1.

Now, we show that the minimum cost maximum flow that we constructed satisfies at least a quarter of the weight threshold demanded by each sink. For each stream with source $k$ and sink $j$,  we know from Lemma~\ref{lem:weightbnd} that $\overline{W}_j^k \geq \frac{3}{4} W_j^k$, with high probability. Recall that for each sink $j$, we created $s_j = \left\lceil 2 \sum_{i \in R}  \bar{x}_{ij}^k \right\rceil$ boxes and potentially eliminated the last box.  In any maximum flow, each (uneliminated) box must receive exactly $1/2$ unit flow so that the edges from level 4 to $\bf{T}$ are saturated. Let $\min(w_{\ell j}^k)$ and  $\max(w_{\ell j}^k)$ denote the smallest and the largest weights respectively assigned to the $l^{th}$ box, for $1 \leq l \leq s_j$.
For the constructed flow, $$\sum_{i \in R} w_{ij}^k \cdot \tilde{x}_{ij}^k  \geq \frac{1}{2} \sum_{\ell=1}^{s_j-1} \min(w_{\ell j}^k),$$ since the last box is potentially eliminated and each uneliminated box receives a flow of $1/2$.  Note that 
\[
\begin{array}{l}
\frac{1}{2} \sum_{\ell=1}^{s_j-1} \min(w_{\ell j}^k)  
\geq \frac{1}{2} \sum_{\ell=2}^{s_j} \max(w_{\ell j}^k) \mathrm{,\ since\ the \ weights\ are\ in\ nonincreasing\ order}\vspace{3mm} \\ 
= \frac{1}{2} \sum_{\ell=1}^{s_j} \max(w_{\ell j}^k) - \frac{1}{2} w_{1j}^k \geq   \sum_{i \in R} w_{i j}^k 
\bar{x}_{ij}^k - \frac{1}{2} w_{1j}^k \geq \overline{W}_j^k - \frac{1}{2} W_j^k \mathrm{,\ since\ } w_{1j}^k \leq W_j^k \vspace{3mm}\\
\geq \frac{3}{4} W_{j}^k  - \frac{1}{2} W_j^k = \frac{1}{4} W_{j}^k\textrm{,\ using Lemma~\ref{lem:weightbnd}}.
\end{array}
\]
Thus, at least a quarter of the weight threshold of each sink is met.

The flow that we have constructed thus far is not 0-1, since some  flow values can equal $1/2$. To rectify this, we double all $\tilde{x}_{ij}^k=1/2$. Clearly, after the doubling we continue to satisfy at least a quarter of  the weight threshold demanded by each sink. However, we might violate fan-out constraints by at most an additional factor of two. Thus, in combination with Lemma~\ref{lem:fan-outbnd}, this means that we meet all  fan-out constraints to within a total factor of at most $4$. We also at most double the  cost $\tilde{C}$ associated with $\tilde{x}_{ij}^k$, but that can be absorbed in the $O(\log{n})$ factor on the cost derived in Lemma~\ref{lem:costbnd}. This concludes the rounding of the last fractional variables of our solution. We get the desired 0-1 solution.  Note that from Theorem~\ref{thm:lowerbound}, we know that the approximation ratio achieved by {\tt Approx} is the best possible to within constant factors.

\subsection{Putting it all together}

We will now calculate the running time of our approximation algorithm {\tt Approx}. First, we will determine the number of variables and constraints in the LP (or, the corresponding IP). Note that we replicated the sources and sinks so that each sink (resp., source) receives (resp., originates) exactly one stream. Therefore, $|S| \leq |D|$ and the total number of variables of the form $x^k_{ij}$ is $|R| \cdot |D|$. 
Thus, the LP (or, IP) has $O(|R|\cdot|D|)$ variables and $O(|R|\cdot|D|)$ constraints.  Since the LP can be solved in time polynomial in the number of variables and constraints, the first step of finding the fractional LP solution takes time polynomial in $O(|R|\cdot|D|)$. The randomized
rounding step takes at most as many iterations as the number of LP variables, so its running time is  $O(|R|\cdot|D|),$ which is dominated by the time for the LP solution step. The  GAP flow graph has $O(|R|\cdot|D|)$ nodes and edges. Thus, the running time of solving the network flow problem on the  GAP flow graph is also polynomial in
$O(|R|\cdot|D|)$.  Thus, {\tt Approx\/} is an efficient algorithm for solving the overlay network construction problem with a run time that is polynomial in $O(|R| \cdot |D|)$.

Putting it all together, we can state the following main theorem.

\begin{theorem}
\label{thm:algfinal}
Algorithm {\tt Approx} solves the overlay network construction problem by constructing paths $\Pi$ for simultaneously routing each stream from its source to its subset of sinks such that at least $\frac{1}{4}^{th}$ of the weight threshold is met for each stream and sink, and the reflector fan-out constraints are met to within a factor of $4$, with high probability.  Further, the expected cost of the solution produced by {\tt Approx} is within a factor of $O(\log n)$ of optimal. {\tt Approx} runs in time polynomial in its input size of $O(|R|\cdot|D|)$. Further, the approximation ratios achieved by {\tt Approx} are the best achievable by any polynomial time algorithm (to within constant factors).
\end{theorem}

Here is some intuition of what the weight guarantee achieved by {\tt Approx} means in our context. Since we
started by converting probabilities into weights using logarithms,  guaranteeing at least $\frac{1}{4}^{th}$ the weight threshold translates to guaranteeing at most the fourth root of the specified end-to-end packet loss  threshold. For example if we want end-to-end packet loss of at most $\Phi_j^k=0.0001$, our solution is guaranteed to provide an end-to-end loss probability of at most $0.1$. Our empirical studies in Section~\ref{sec:implement} indicate, however, that the extent of weight constraint violations can be much less in practice than the theoretical guarantees provided above.

 \section{Extensions and modifications}
\label{sec:extensions}

In this section, we examine extensions and modifications of the overlay network construction problem that are relevant for practical applications.

\subsection{Minimizing transmission costs}
Perhaps the most important special case of the overlay network construction problem is the common situation where the reflectors are considered  to be``free'' and the operating cost is entirely dictated by the bandwidth costs for transmitting streams from their sources to their respective sinks.  In this formulation, the fixed cost of utilizing a reflector is considered ``sunk'' cost and the overlay network is periodically reconstructed  to minimize transmission costs while obeying capacity constraints and maintaining quality of service. To model this situation, we can set the cost $r_i = 0$, for all $i \in R$. Further, our cost objective function can be simplified to 
\[ \sum_{i \in R}\sum_{ k \in S} \sum_{ j \in D} C_{ij}^k x_{ij}^k, \]
where $C_{ij}^k$ captures the entire bandwidth cost of transmitting  the stream from its source $k$ to sink $j$ via reflector $i$, i.e., 
\[C_{ij}^k = c_{ki}^k  + c_{ij}^k.\]
Note that since we are no longer considering the capital expenditure cost of purchasing reflectors, the overlay construction problem no longer contains set cover as a special case, suggesting perhaps that better approximation ratios are possible. In fact, we now show that our algorithm {\tt Approx} achieves an approximation ratio of $2$ for the above simplified cost objective function, rather than the approximation ratio of $O(\log n)$ for the general case. To see why, let $C^{OPT}$ be the optimal transmission cost. From Equation~\ref{eq:xbar}, we conclude that the expected cost after the randomized rounding step  equals the cost of the LP solution, which in turn is at most $C^{OPT}$. The GAP rounding step increases the cost by at most a factor of $2$, hence the cost of the solution produced by {\tt Approx} is at most $2 C^{OPT}$. Thus, we can state the following theorem.
\begin{theorem}
\label{thm:algtransmission}
In the special case where transmission costs are minimized, Algorithm {\tt Approx} produces a solution with expected cost that is within a factor of $2$ of optimal. Further, at least $\frac{1}{4}^{th}$ of the weight threshold is met for each stream and sink, and the reflector fan-out constraints are met to within a factor of $4$, with high probability.  
\end{theorem}

Is it possible to achieve an even better approximation ratio in the special case where only transmission costs are minimized? We show that the overlay network design problem for minimizing transmission cost is still $\NP$-hard. 
However, our argument does not forbid the existence of better approximation schemes for the problem, for example a PTAS,  which remains open. 

\begin{theorem}
\label{thm:NPHardtransmission}
The overlay network design problem for minimizing transmission cost is NP-Hard.
\end{theorem}
\begin{proof}
We show that a simple restriction of the overlay network design problem with zero reflector costs yields the subset sum problem that is $\NP$-hard~\cite{GJ79}. Suppose that we have just a single source $k$ that originates a single stream and just two sinks $A$ and $B$ that demand that stream. Further, let each reflector have a fan-out constraint of one (i.e., it can serve only a single sink).  In addition, assume that for each reflector $i$, the weights $w^k_{i,A} = w^k_{i,B}$, i.e., each reflector has the same weight to either sink. Now the question of  ``Is there a way to provide weight at least $W^k_A$ to sink $A$ and
at least weight $W^k_{B}$ to sink $B$ without violating any fan-out constraints?'' is equivalent to  ``Is there
a way to partition the set of reflectors into two groups such the sum of the weights of the first group assigned to $A$ is at least $W^k_A$ and the sum of weights of the second group assigned to $B$ is at least $W^k_B$?''. The latter question is equivalent to  the subset sum problem that is  $\NP$-hard. However, it is worth pointing out that the subset sum problem can be solved in polynomial time provided
that there are only a constant number of distinct values for the weights. It also has a FPTAS (Fully-Polynomial Time Approximation Scheme) for any number
of distinct weight values \cite{CormenLRS2009}.
 \end{proof}
\subsection{Bandwidth capacity of reflectors}
\label{sec:bandonreflectors}     
An important constraint in practice is to bound the aggregate bandwidth (in bits per second) that a reflector  can push  to the sinks, rather than just the fan-out. This capacity bound is due to both hardware and software limitations of the reflector machines. We can introduce capacity bounds  by introducing the following constraints (3') and (4') to the IP formulation of Section~\ref{subsec:ipf}:
\[ 
\begin{array}{l}
(3') \hspace{10mm} \sum_{k \in S} \left (B^k \cdot \sum_{j \in D} x_{ij}^k \right) \leq F'_i z_i \; \; \forall i \in R \vspace{3mm}\\ 
(4') \hspace{10mm}  B^k \cdot \sum_{j \in D} x_{ij}^k \leq F'_i y_i^k \; \; \forall i \in R,\; \forall k \in S
\end{array}
\]
Here $B^k \in {\R^{+}}$ can be viewed as the encoded bandwidth (in bits per second) for the stream that originates at source $k$. Thus, the LHS of constraint (3') equals the total bandwidth sent by reflector $i$ and $F'_i$ is the bandwidth capacity bound that represents the maximum bandwidth (in bits per second) that reflector $i$ can send.  In this case, with small modifications, both our algorithm and our analysis hold. Specifically, a very similar argument to the one that we used to bound fan-out can be used to bound the bandwidth capacity instead.

%\subsection{Capacities on all of the arcs}

%Now we consider a capacitated version of the problem, i.e. we add new constraints
%\[ 
%\begin{array}{l}
%(7) \hspace{10mm} \sum_{k \in S} x_{ij}^k \leq u_{ij} \; \; \forall i \in R, \forall j \in D \vspace{3mm} \\
%(8) \hspace{10mm} \sum_{k \in S} y_i^k \leq u_i \; \; \forall i \in R
%\end{array}
%\]
%Here
%$$u: \mbox{E} \rightarrow {\R_{+}}.$$
%If we assume that there exists a randomized algorithm which solves this modification of the problem 
%by violating constraints (7) and (8) with a constant factor, then we showed there will be an algorithm
%that approximates {\em Set cover} to within a constant factor. Since the latter is highly unlikely~\cite{Feige98}
%there is not much hope for an interesting solution to this version of the problem. Note that our 
%rounding procedure described previously, applied to a fractional solution of the LP relaxation of the modified problem, 
%will yield a $c\log{n}$ factor violation of constraints (7) and (8) - the best guarantee we can hope for.

%\subsection{Bounding the number of streams between reflectors and sinks}

%We consider constraints that bound the number of distinct streams that can be sent between reflectors and sinks that can be added to the IP formulation of Section~\ref{subsec:ipf}:
%\[  (7') \hspace{10mm} \sum_{k \in S} x_{ij}^k \leq u_{ij},\]
%where $u_{ij} \in {\R^{+}}$,  $\forall i \in R, \forall j \in D$.
%The solution to this extension is similar to that for color constraints and is hence discussed in Section~\ref{sec:colconstraints} below.

\subsection{Enhancing reliability for correlated failures}
\label{sec:colconstraints}
Data centers hosted on the same ISP are more likely to fail simultaneously in a correlated fashion than data centers hosted on different ISPs. Such a failure is typically caused by some catastrophic event impacting that ISP. In the worst case, such a failure can render machines hosted in the failed ISP's data centers unreachable from rest of the Internet. Therefore, in a situation where there is a need to employ multiple paths between a source and a sink using multiple reflectors, we would prefer to use reflectors located in as many {\em distinct\/} ISPs as possible\footnote{ISP failures also impact entry points and edge servers and can be tackled through other means. One can build in fault tolerance for the entry points by automatically reassigning an alternate entry point to avoid the failed one  \cite{KontothanasisSWHKMSS04}. Further, one can move end users away from the failed  edge servers using mapping \cite{NygrenSS10}. In this paper, we focus only on the impact on reflectors in the ``middle-mile''.}. In other words, we would like to restrict the number copies that a sink receives from reflectors hosted on the same ISP. This provides an additional level of fault tolerance against coordinated failures that impact the transmission of the stream from a source to its sink. We can model this additional constraint by assigning colors to each reflector such that a reflector's color represents the ISP where it is hosted. That is, we partition $R$ by color into disjoint sets so that $R=R_1 \cup R_2 \cup R_3 \ldots \cup R_m$, where $m$ is total number of colors. 
Then, we have the following ``color constraints'' added to the IP formulation of Section~\ref{subsec:ipf}:
\[ (7) \hspace{10mm} \sum_{i \in R_{\ell}} x_{ij}^k \leq 1. \; \; \forall j \in D,\; \forall k \in S, \; 1 \leq \ell \leq m.\]
The purpose of these constraints is to break the reflectors into disjoint groups and ensure that no group is delivering more than one copy of the 
stream into a sink. 

We incorporate the additional color constraints in our algorithm for constructing overlay networks as follows. We first solve the LP relaxation with the color constraints. As before, we perform two steps of rounding to obtain an integral solution from the fractional LP solution. The first randomized rounding step in 
Section~\ref{sec:randround} can be carried out with no modifications. However, the final step of rounding
$\bar{x}_{ij}^k$ using the GAP flow graph in Section~\ref{sec:gap} requires modification. The color constraint restriction
introduces a new type of constraint in the GAP flow graph (Figure~\ref{fig:conversion}). This constraint bounds the total flow along some subsets of the edges between the second and third level of the GAP flow graph. 
Such constraints can be introduced into any flow problem. On a general graph this problem is called the ``Integral Flow
with Bundles Problem'' and is known to be $\NP$-hard~\cite{GJ79}.
A key issue that makes the problem more complex is that the introduction of a color constraint can create a gap between the optimal fractional and integral flows,
even in the more restricted leveled graph case that we are interested in. We provide a simple example in Figure~\ref{SetConstraint} to demonstrate this point. The capacities for all edges are as shown in the Figure~\ref{SetConstraint}. 
Suppose there is an additional set constraint that the set of edges $\{AB,PQ\}$ has
a capacity of 3.
\begin{figure}[t]
        \centerline{\psfig{figure=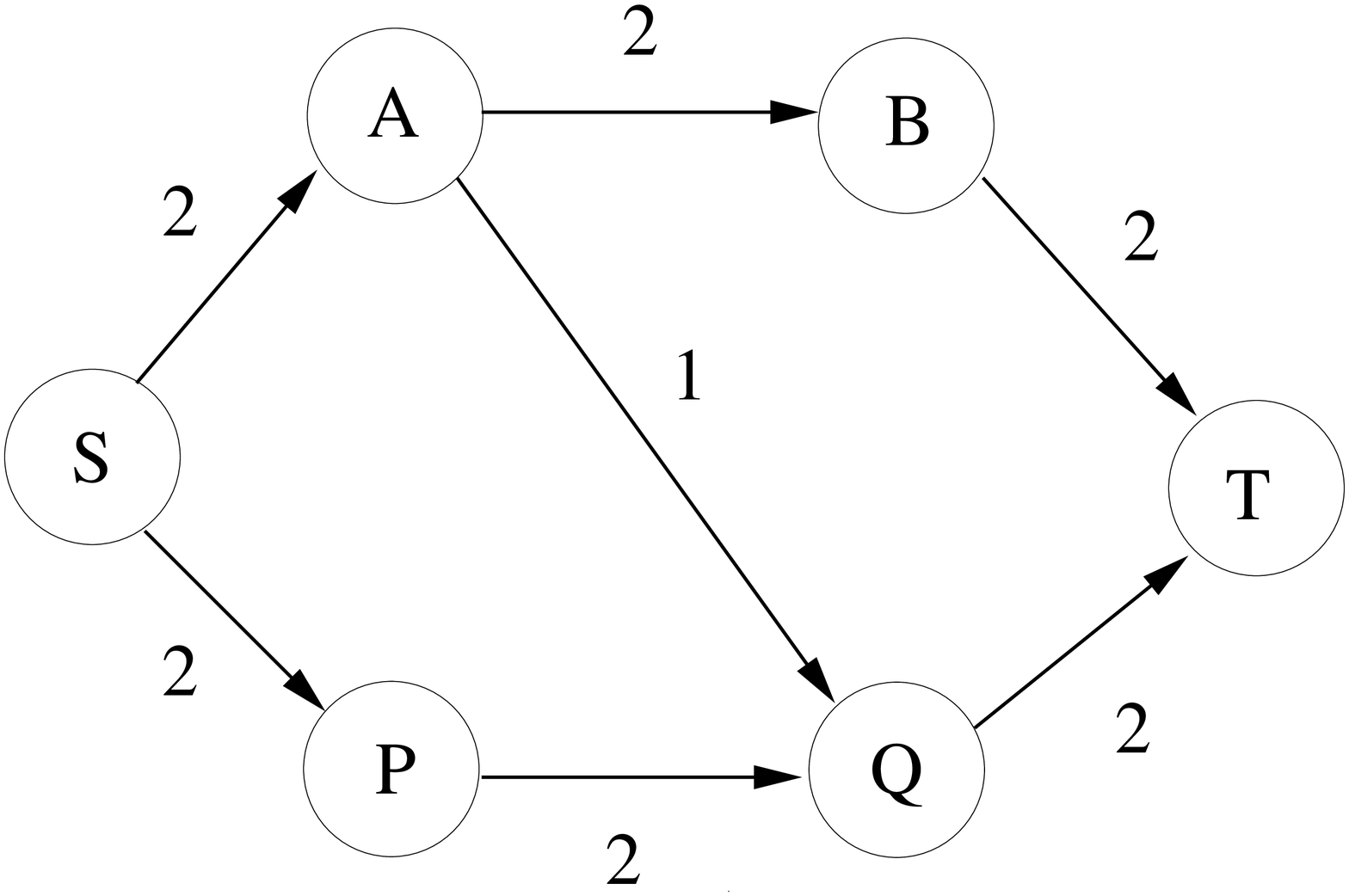,height=2.2in,width=3.2in}}
        \caption{Example of an integral flow with bundles problem.}
\label{SetConstraint}
\end{figure}
Clearly the max integral flow is only 3. However one can achieve a
fractional max flow of 3.5 units, by sending 2 units of flow on $SA$
and 1.5 units on edge $SP$ then splitting the flow at $A$ by sending .5
units on edge $AQ$ and the rest on $AB$.  This phenomenon will prevent us from
applying GAP directly, as we cannot always find an integral flow that is at
least as good as the fractional flow, necessitating a different approach.

Our approach finds an integral solution within a constant factor (of at most 13) of optimal cost while violating the constraints by an
additional constant factor (of at most 13) by adapting techniques due to Srinivasan and Teo~\cite{SriTeo01} and using an LP rounding theorem due to Karp et 
al.~\cite{KLRTVV87} . Given the larger constants, we view our results for enforcing color constraints to be primarily of theoretical interest. We leave open both the practical evaluation of these techniques as well as better algorithms with provably smaller constants.

We reformulate the flow problem on the GAP flow graph (see Figure~\ref{fig:conversion}) as a new LP in terms of paths. In the GAP flow graph, let $\mathcal{B}$ be
the set of boxes (nodes) at level 4  and let $\mathcal{P}$ be the set of all
paths  from $\bf{S}$ to the boxes in $\mathcal{B}$. Further, for each $1 \leq l \leq m$ and $j \in D$, let $S_{l, j}$  be the set of all edges from a node in level 2  to a node in level 3 such that the node in level 2 represents some reflector $i \in R_l$ and the node in level 3 represents the reflector-sink pair $(i,j)$.  Finally, let the variable $\pi_p \in [0,1]$ indicate the amount of fractional flow carried by path $p$, for each $p \in \mathcal{P}$.  The LP follows.
\[
\begin{array}{ll}
(i) \hspace{10mm} &{\displaystyle \sum_{p \in {\mathcal P}| e \in p}
\pi_p \leq u_e,} \hspace{5mm} \mathrm{\forall \; e \in E, where\ E\ is\ the\ set\ of\ edges\ in\ the\ GAP\ flow\ graph}\\
(ii) \hspace{10mm}&{\displaystyle \sum_{p \in {\mathcal P}| p = \{{ \bf{S}}
\rightarrow b\}} \pi_p = \frac{1}{2},}
\hspace{5mm} \forall \; b \in \mathcal{B}\\
(iii) \hspace{10mm}&{\displaystyle \sum_{p \in {\mathcal P}| p \cap S_{l,j} \neq \emptyset} \pi_p \leq 1}, \hspace{5mm}   j  \in D \ and\ 1 \leq l \leq m, where\ m \ is\ the\ number\ of\ colors \\
(iv) \hspace{10mm}&{\displaystyle \sum_{p \in {\mathcal P}} c_p
\pi_p \leq C}
\end{array}
\]
Here $u_e$ is the capacity on edge $e \in E$, ${\bf{S}}$ is the node in the level 1 of GAP flow graph, $\{{\bf{S}}
\rightarrow b\}$ denotes a path from $\bf{S}$ to a box $b$, $c_p$ is the cost of path $p \in
\mathcal{P}$, and $C$ is the total cost of the solution produced by the randomized rounding step.
The constraints $(i)$ above codify the capacity constraints on the edges.
Constraints $(ii)$ require a flow of $1/2$ to each of the boxes in $\mathcal{B}$. Constraints $(iii)$ are the special (set-type) color constraints and 
constraint $(iv)$ controls the cost. 

As we saw earlier, the values $\bar{x}_{ij}^k$  obtained from the randomized rounding step can be used to create a valid flow on the GAP flow graph. One can decompose this flow into flow paths in the standard fashion and produce a feasible fractional solution for the above LP that we denote by $\bar{\pi}_p, p \in \mathcal{P}$.  Next, analogous to Srinivasan and Teo's technique, we do a step of path filtering to eliminate all ``expensive'' paths $p$ such that $c_p > 4C$, resulting in a smaller set of paths $\mathcal{P'} \subseteq \mathcal{P}$.  Using the fact that $c_p > 4 C$ for $p \in {\mathcal P - P'}$, we have
\begin{equation}  {\displaystyle \sum_{p \in {\mathcal P - P'}} \bar{\pi}_p < \frac{1}{4},}  \label{eq:lessthan4}
\end{equation}
since otherwise the total cost of the solution would be more than $C$, leading to a contradiction.
 Thus, using constraint (ii) of the above LP and Equation~\ref{eq:lessthan4}, we have
\begin{equation}
 {\displaystyle \sum_{p \in {\mathcal P'}| p = \{{ \bf{S}}
\rightarrow b\}} \bar{\pi}_p \geq \frac{1}{4},} \hspace{5mm} \forall \; b \in \mathcal{B} \label{eq:pathfiltering}
\end{equation}
Now, set $\tilde{\pi}_p = 4 \bar{\pi}_p$, for all $p \in \mathcal{P'}$. The values $\tilde{\pi}_p$,  $p \in \mathcal{P'}$, are a feasible solution to the following LP, where constraints (i), (iii), and (iv) below are obtained by quadrupling the RHS of the corresponding constraints of the prior LP. Further, constraints (ii) below are obtained from Equation~\ref{eq:pathfiltering} and then multiplying both RHS and LHS by negative 36 for reasons that will become clear when we apply Theorem~\ref{thm:karpetal} below.
\[
\begin{array}{ll}
(i) \hspace{10mm} &{\displaystyle \sum_{p \in {\mathcal P'}| e \in p}
\pi_p \leq 4 u_e,} \hspace{5mm} \mathrm{\forall \; e \in E, where\ E\ is\ the\ edge\ set\ of\ GAP\ flow\ graph}\\
(ii) \hspace{10mm}&{\displaystyle   \sum_{p \in {\mathcal P'}| p = \{{ \bf{S}}
\rightarrow b\}} -9  \pi_p \leq -9,}
\hspace{5mm} \forall \; b \in \mathcal{B}\\
(iii) \hspace{10mm}&{\displaystyle \sum_{p \in {\mathcal P'}| p \cap S_{l,j} \neq \emptyset} \pi_p \leq 4,} \hspace{5mm}\;  1 \leq l \leq m \ and\  j  \in D\\
(iv) \hspace{10mm}&{\displaystyle \sum_{p \in {\mathcal P'}} \left(\frac{c_p}{C}\right)
\pi_p \leq 4}
\end{array}
\]

Now we round the fractional solution $\tilde{\pi}_p$ to obtain an integral solution using the following result due to Karp et al.~\cite{KLRTVV87}. 
\begin{theorem}[\cite{KLRTVV87}]
\label{thm:karpetal}
 Let $A$ be a real valued $r \times s$ matrix and $z$ be a real-valued
$s$-vector. Let $b$ be a real-valued vector such that $Az = b$ and $t$ be a positive real number
such that, in every column of $A$, (i) the sum of all the positive entries is at most $t$
and (ii) the sum of all the negative entries is at least $-t$. Then it is possible to compute an
integral vector $\ddot{z}$ such that for every $i$, either $\ddot{z}_i =\left\lfloor z_i \right \rfloor$ or $\ddot{z}_i = \left\lceil z_i \right \rceil$ 
 and $A\ddot{z} = \ddot{b}$ where
$\ddot{b}_i - b_i < t$ for all $i$. Furthermore, if $z$ contains $d$ nonzero components, the integral
approximation can be obtained in time $O(r^3 \log(1 + s/r) + r^3 + d^2r + sr)$.
\end{theorem}

To use the above theorem, note that our inequalities can be converted into equalities by using the standard trick of adding a distinct slack variable for each constraint \cite{CormenLRS2009}. Next, 
we bound the sum of the positive coefficients for each
$\pi_p$ in the above LP.  The variable $\pi_p$ appears 4 times
(at most once for each level) in constraints $(i)$ with a coefficient of $1$, at most once in $(iii)$ with a coefficient of $1$, and
exactly once in $(iv)$ with coefficient at most 4. This adds up to a
total of $9$. Further, each slack variable appears only in one constraint with a coefficient of 1. Likewise, the negative coefficients for any $\pi_p$ from constraints (ii) is at least $-9$. Applying Theorem~\ref{thm:karpetal}, we round the feasible fractional solution $\tilde{\pi_p} \in [0,1]$ to the above LP to get an integral solution $\ddot{\pi_p} \in \{0,1\}$ that
satisfies all the constraints with an additive factor less than 9. Thus, the rounded values $\ddot{\pi}_p$ satisfies the following modified constraint (ii):
\[
\sum_{p \in {\mathcal P'}| p = \{{ \bf{S}}
\rightarrow b\}} -9  \pi_p < -9 + 9 = 0, \hspace{5mm} \forall \; b \in \mathcal{B}
\]
The strict inequality in the above equation is important since it guarantees that there is at least one path $p \in {\mathcal P'}$ from source $\bf{S}$ to each box $b$ with $\pi_p = 1$ that can be used for routing. Further, this additive
factor translates into an approximation ratio of $4 + 9$ equal to $13$ for the cost, i.e., the obtained cost is no more than a factor of $13$ from optimal. Finally, the coloring constraints are (approximately) satisfied to ensure that no more than $4 + 9 = 13$ copies of any stream are sent to a sink from reflectors belonging to a particular ISP. Thus
we get the promised approximation guarantees, though the larger constants for satisfying the coloring constraints make the result of theoretical significance only, as streams in practice seldom use more than 3 paths in total to meet their packet loss thresholds.

The running time of the rounding step can be evaluated by observing that the number of non-zero 
values of $\tilde{\pi}_p$ (which is at most $|\mathcal{P}|$), the number of constraints in the LP, and the number of variables in the LP are each $O(|R| \times |D|)$. 
Thus, applying Theorem~\ref{thm:karpetal} the running time is $O(|R|^3\cdot|D|^3)$.
  
\section{Implementation and Experimental Results}
\label{sec:implement}
In this section, we demonstrate the efficacy of the approximation algorithm outlined in this paper by implementing and running it on realistic inputs derived from Akamai's live streaming network. We implemented our algorithm {\tt Approx} in C++.  As a comparison, we also implemented two other algorithms that we call {\tt ApproxHack} and {\tt IP}, resulting in three different algorithms being compared. The computers that were used to run the experiments each had a single
Intel Pentium 4 processor clocked at 2.4 Ghz and had 1 GB of RAM. 
When we compare two solutions from the different algorithms for the same input,
we ensure consistency by always running the experiments on the same machine.
\begin{figure}[h]
\begin{center}
$\begin{array}{c@{\hspace{.1in}}c}
\psfig{figure=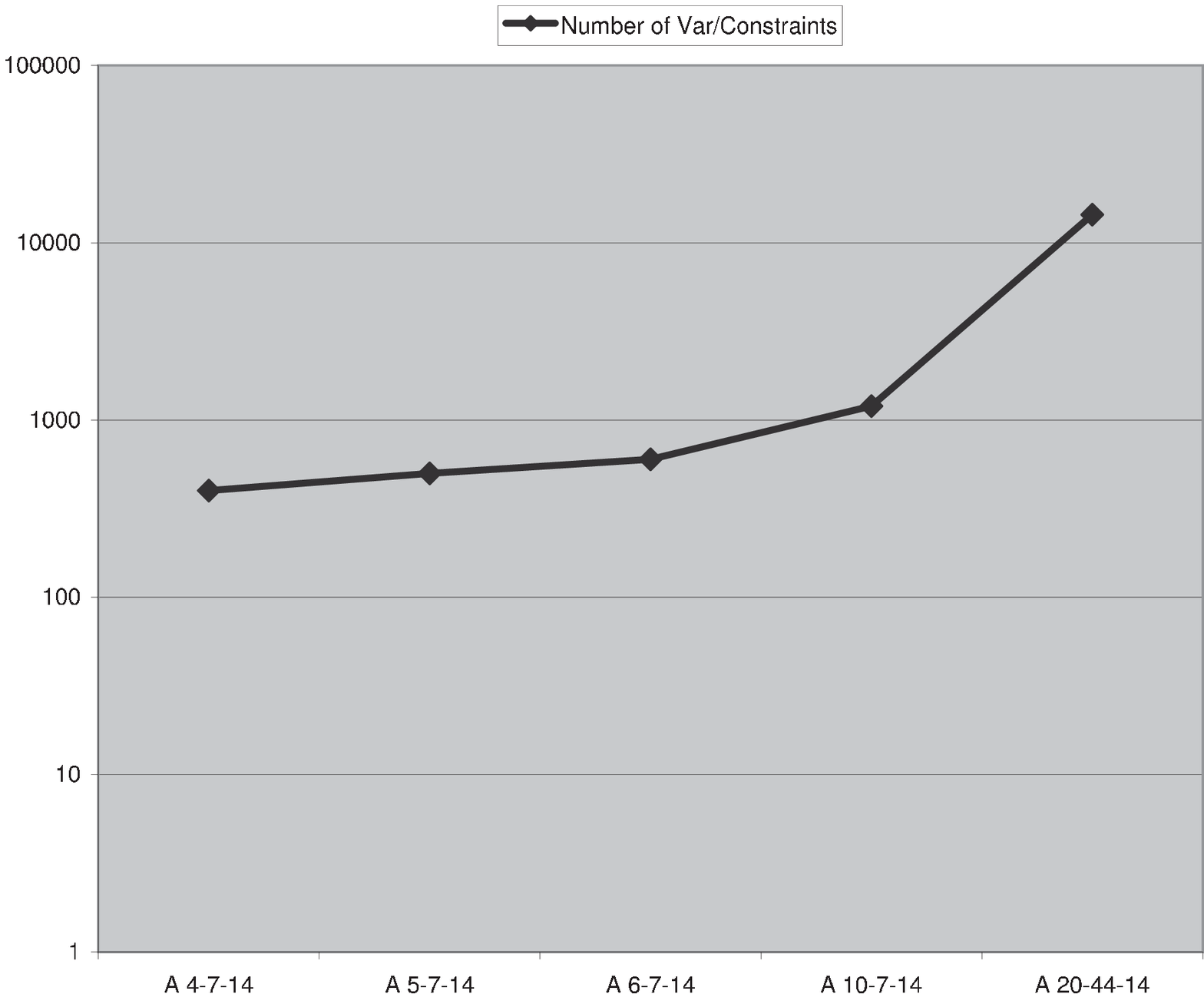,angle=90,width=3.25in} &
\psfig{figure=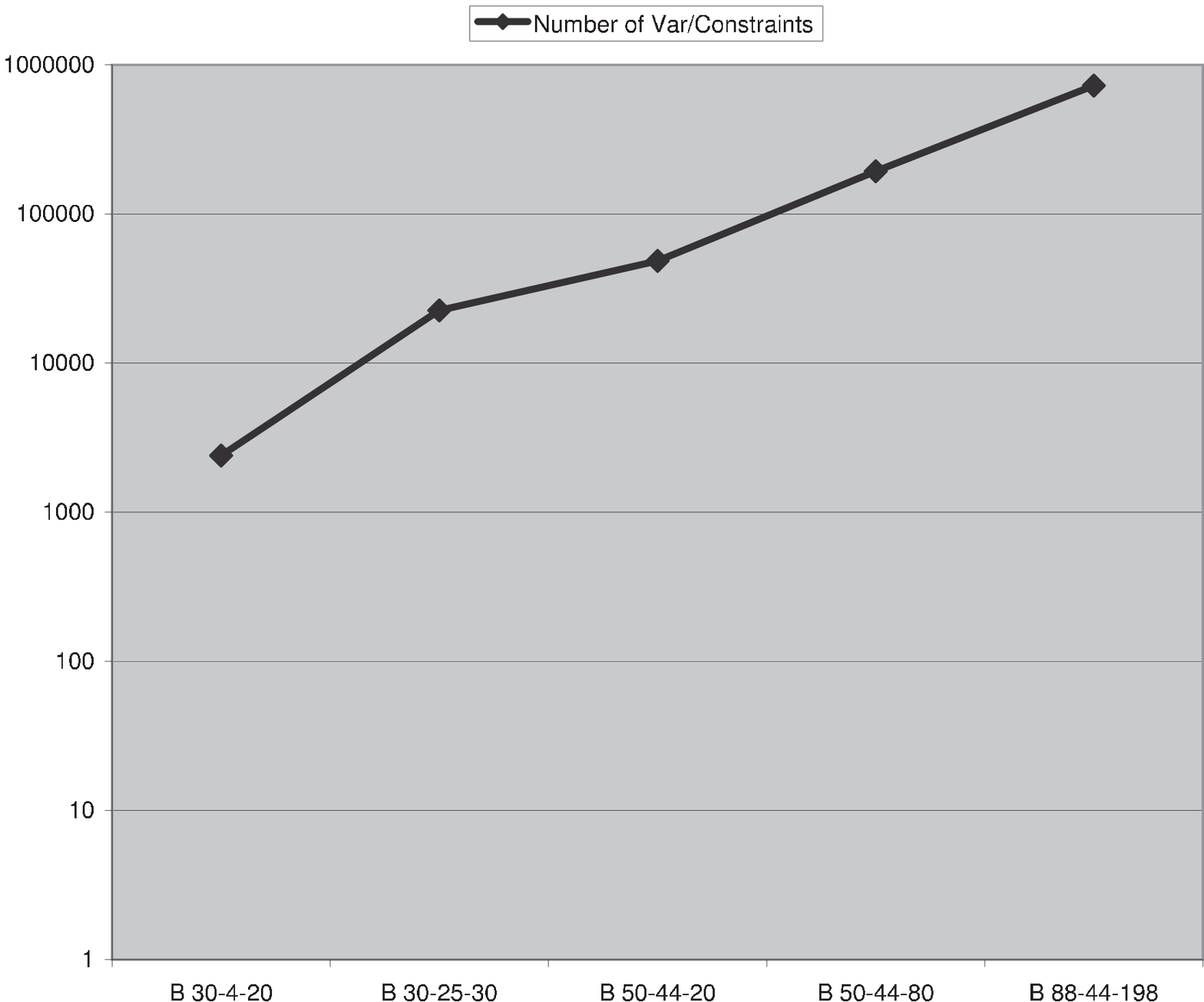,angle=90,width=3.25in}
\end{array}$
\end{center}
\begin{center}
\psfig{figure=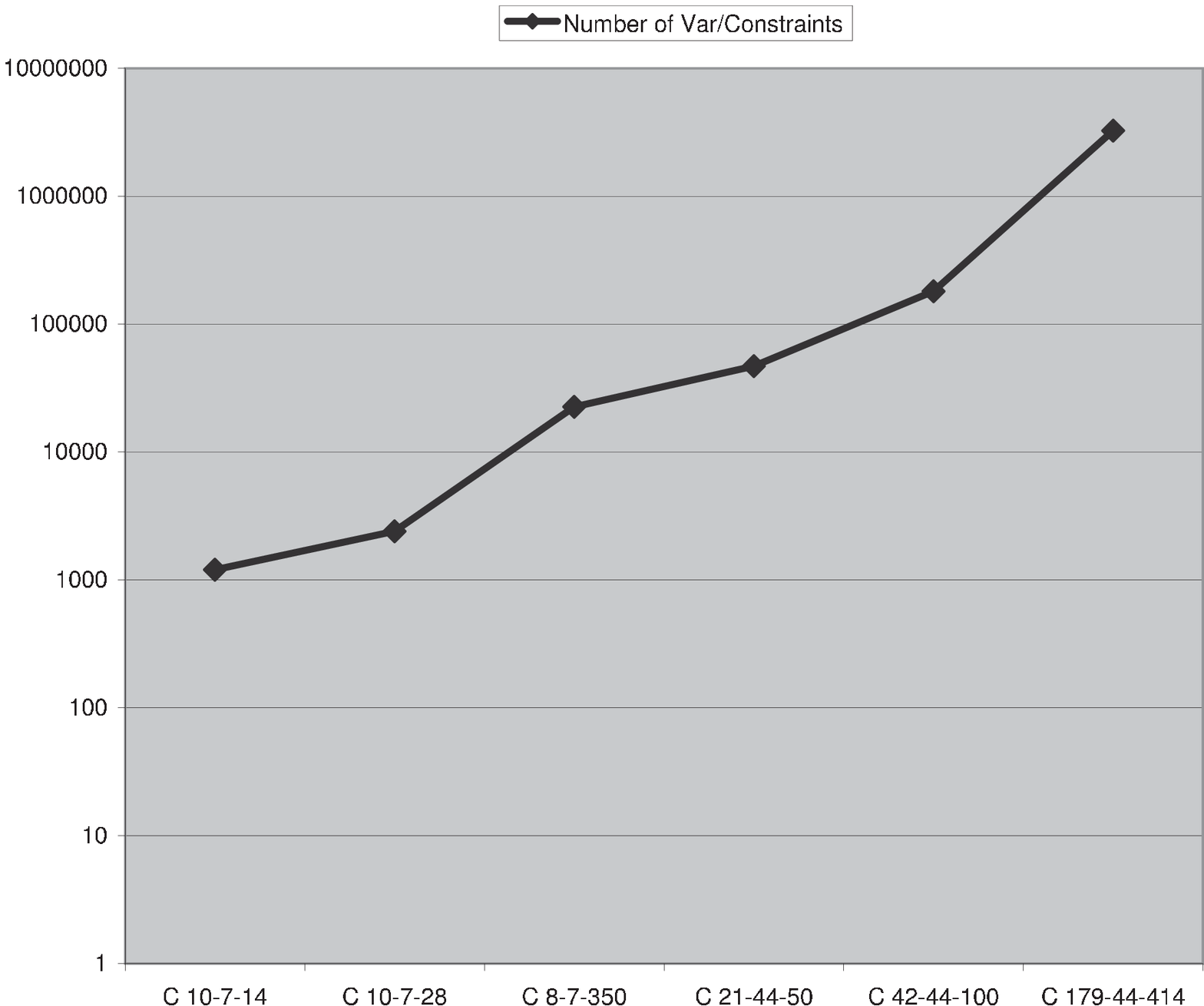,angle=90,width=3.25in}
\end{center}
\caption{Problem Size and Network Sizes for Media Formats}
\label{Network Sizes}
\end{figure}
\begin{figure}[h]
\begin{center}
$\begin{array}{cc}
\epsfig{figure=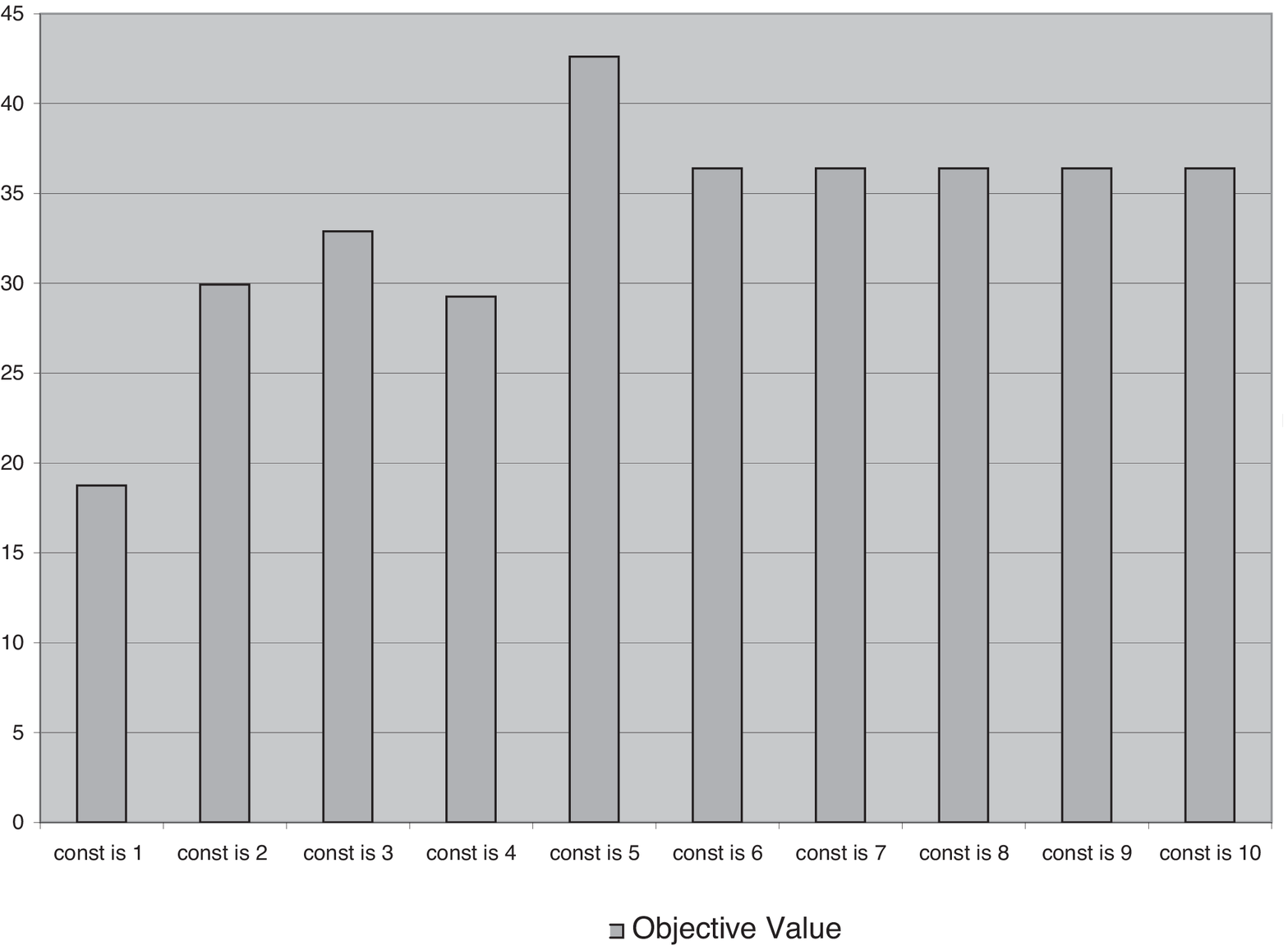,width=3.35in}&
\epsfig{figure=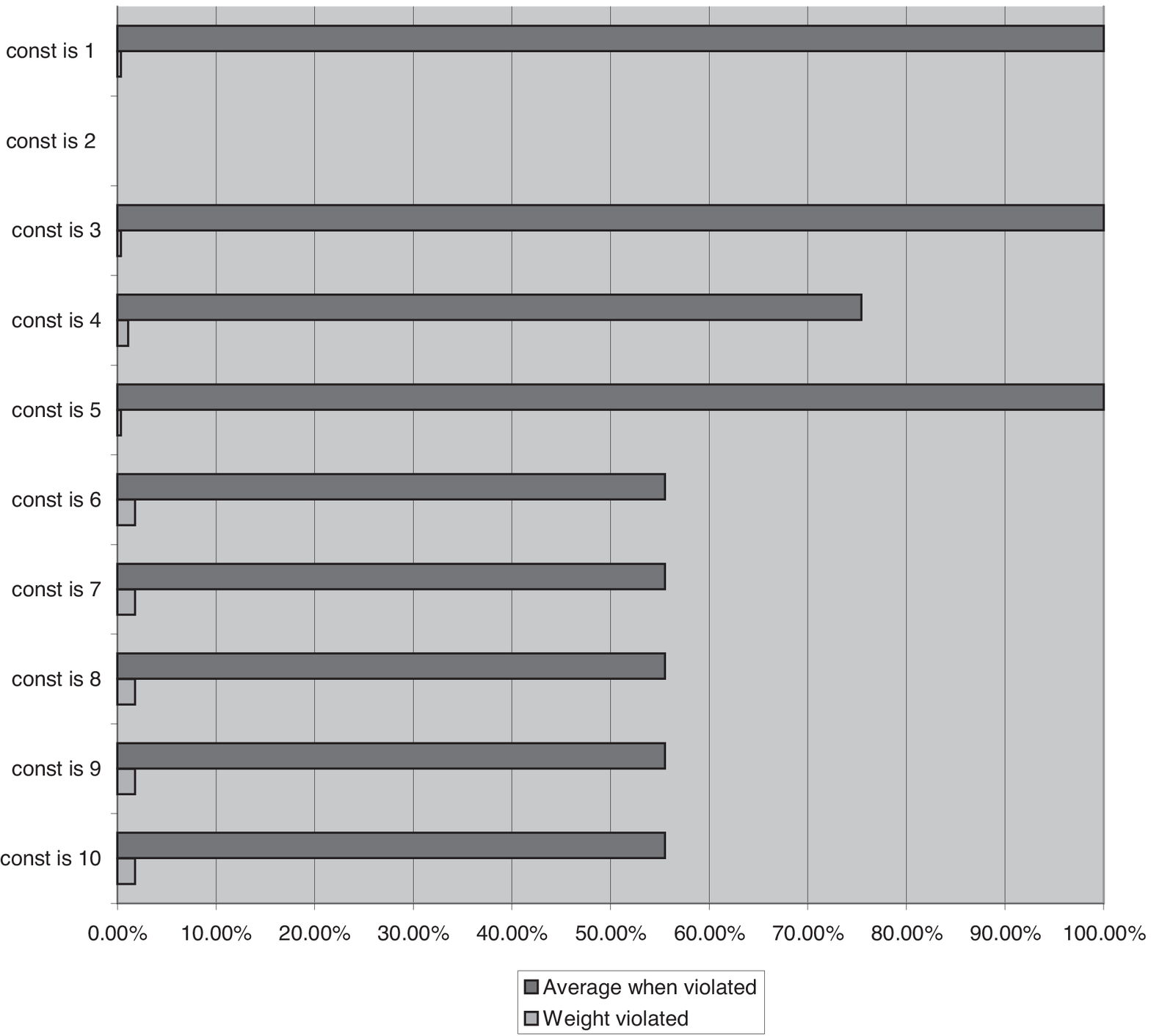,width=3.35in}\\
\end{array}$
\end{center}
\caption{ Cost objective value and loss violations for small constant multipliers for a $10\times7\times28$ network under average loss conditions for format C.}
   \label{C 10-7-28 ObjViol}
\end{figure}

%We paid special attention to run the same input for the two algorithms 
%we compare on the same machine.

As we saw earlier, {\tt Approx}  consists of solving a linear program followed by randomized rounding and another rounding step using modified GAP approximation. We now propose a local-search variant of our solution called {\tt Approxhack} where the linear program is solved. But, rather than perform a rounding process, we fix any variables that turn out to be integral in the LP solution, and solve the remaining variables using integer programming. One can think of {\tt Approxhack} as performing a local search heuristic as it tries to find an optimal
integer solution within the neighborhood of the fractional LP solution.
Note that {\tt Approxhack} always produces a solution with cost that is at most the cost of the solution produced by {\tt Approx}.  The reason is that   both {\tt Approx} and {\tt Approxhack} do not alter variables that happen to have integral values in the LP solution and leave them as is in the final solution. Thus, {\tt Approxhack} that solves an integer program for the remainder of the variables produces a solution that is at least as good as any other way of determining the values of those remaining variables, including the rounding procedure used by {\tt Approx}. However, unlike {\tt Approx}, {\tt Approxhack} does not run in polynomial time as it involves solving an integer program that could take exponential time. Finally, we implement a third solution that we refer to as {\tt IP}, that simply solves the integer programming formulation directly, without using  linear programming relaxation or rounding. {\tt IP} always produces the optimal solution with the smallest cost, but could take exponential time as it solves an integer program. Note that the cost achieved by {\tt IP} is at most the cost achieved by {\tt ApproxHack}, which in turn is at most the cost objective value achieved by {\tt Approx}. However, {\tt Approx} is the only polynomial time algorithm of the three.

For solving linear and integer programming problems,
we investigated using a number of packages, including 
DashOptimization XPress-MP, AMPL (plugging in any number of the supported 
solvers), Ciplex's Concert API, and GLPK. We also used COIN-OR's (http://www.coin-or.org) 
Open Solver Interface to call the different mathematical programming solvers. 
In general in our experiments it turned 
out that the running time of the solvers do not differ significantly so we 
chose to use GLPK as our default solver.  Hence GLPK is used as the solver in all three algorithms to solve linear and integer programs as needed.

 For randomness we used srand seeded with the current time and rand()/RAND\_MAX 
to generate random numbers between 0 and 1. We also used Boost's Graph Library 
(http://boost.org) for implementing the modified GAP rounding as a part of our algorithm {\tt Approx}.

\subsection{Input}

We collected usage data from Akamai's live streaming network to make our input as realistic as 
possible. Recall that Akamai's live streaming network is a three-layered network of entry points (sources), reflectors, and edge servers (sinks) as assumed in our work. To protect the privacy of end-users and content providers, and to remove proprietary Akamai information,
we erased server and stream names and we normalized costs. Cost data was challenging
to compile because some of the contractual agreements between Akamai and the various 
Internet service providers are complicated and 
have an array of clauses. Nevertheless we were able to come up with numbers that
represent a good estimate of the actual cost incurred.

Next, we used actual Internet loss data as measured from the Akamai network. The loss data was collected by Aditya Akella and Jeff Pang. For more on the methodology
of collecting the data, please refer to~\cite{FABK03,APMSS04}. 
We observed that there are two distinct 
periods with respect to the losses, a low loss period that occurs during the night and on
weekends, and a high loss period between 12 noon EST and 6 pm EST when Europe and both coasts of North America are active. We ran our algorithms on three different sets of loss data. One set is representative of a low loss period, another is representative of a high loss period, and a third represents an average loss situation that is in between the first two. The low and high loss data were extracted directly from the Akamai traces. To simulate an average loss case, we averaged the Akamai traces over a full 24 hour cycle.

 The Akamai live streaming network can be viewed as a collection of deployed networks, one for each format. Therefore, we broke the entry points (each of which correspond to one or more sources that originate  one or more streams) into three groups
based on format: Windows Media (WMS), Real Media (REAL), and Quick Time (QT).
We anonymized them by calling them media formats A, B, and C. For creating smaller networks, we removed entry points, reflectors and/or edge servers from the above deployed media format networks.
We tested a wide variety of deployed networks starting from a $4\times7\times14$
network all the way up to $179\times 44 \times 414$ (here the first number represents the number of entry points, the second one is number of reflectors and the third one is the number of edge servers).

\subsection{Experimental results}
We ran our algorithms for each scenario of average loss, low loss, and high loss.
Within each of these three cases, we study all three media formats as well as deployed networks of different sizes. For each media format and deployed network, we routed a set of streams to their respective sinks using {\tt Approx}, {\tt ApproxHack} and {\tt IP}. The streams and their demand patterns are derived from usage traces from Akamai's live streaming network. The size of the problem that we solve is the number of variables in the corresponding integer programming formulation of the problem in Section~\ref{subsec:ipf} and is $O(|R| \times |D|)$. In Figure \ref{Network Sizes}, we show the problem sizes that we solve for each format and deployed network. The x-axis is labeled by the media format (A, B, or C) followed by the deployed network characteristics (number of entry points, reflectors, and edge servers). The y-axis is a log-based plot of the corresponding problem sizes.

\subsubsection{Setting the multiplicative factor}

We want to point out an interesting practical feature of our algorithm {\tt Approx} that we observed in our experimentation. Recall that we have a preset multiplier of $c\log n$ that is used in step 1 of the randomized rounding procedure in Section~\ref{sec:randround}. This multiplier influences how close the cost of the solution is to the optimal as well as the degree to which the weight constraints are met. Specifically, choosing a smaller multiplier will bound the cost objective function to be closer to the optimum value. However, a smaller multiplier will also make it more likely that the weight constraints are violated to a larger degree. 
In our experiments, we explored this tradeoff by varying the value of the multiplier and running our algorithm on various network sizes and loss scenarios.  

Our key finding is that we can achieve very good results even with small constants as multipliers, which implies that {\tt Approx} produces a better approximation in practice than the theoretical bounds imply. For example, we ran experiments on media format C under average loss conditions on  a $10\times 7 \times 28$ deployed network. We tried all multipliers starting from 1 to 280. We found that
in the beginning the objective values varied a lot from run to run, but with the multipliers larger than 6 the solution stabilizes, and we get the same solution all the time.  This shows that in reality, the multiplier can be set to small value and the cost of the solution produced by {\tt Approx} is no more than a small
constant away from the optimal solution. The results from this experiment are summarized in Figure~\ref{C 10-7-28 ObjViol}, where we show both the cost objective value obtained (smaller is better) and the average constraint violation.
Note that in this specific example network we happened to get a better solution with a constant of 2 than the one the algorithm eventually converged to. In Figure~\ref{C 10-7-28 ObjViol}, one can see that it has a better cost objective value and no violated constraints.  However, this solution did not recur in a stable fashion for higher values of the constant.

%Another important factor in our algorithm is the randomness. It is most apparent
%for randomized rounding when we set the multiplier to equal to %1. In this case 
%we get different cost objective values from run to run. Using the %same test set, C 10-7-28, we run
%the algorithm with constant equal 1. We determined that the %expected value
%for the objective is 28.5 and the standard deviation is 5.07.

\subsubsection{Behavior under different loss conditions}
We now compare our algorithm {\tt Approx} with {\tt Approxhack} and {\tt IP} described earlier in terms of run time and the cost achieved. We simulated all three algorithms for different sized networks, different formats, and different loss conditions.  We expect our algorithm {\tt Approx} to be the fastest as it is the only polynomial time algorithm. Note that both {\tt Approxhack} and {\tt IP} involve solving an integer program that could take exponential time and could become infeasible for larger problem sizes. However, we expect {\tt IP}  to produce the smallest value for the cost objective function, followed by {\tt Approxhack} as the next best. 

In addition to  these three algorithms, we also plot the value of the cost objective function of the fractional solution to the LP. Note that the cost of the LP solution is a lower bound on the cost produced by {\tt IP}. Though the LP solution may not always be a valid solution for the overlay network construction problem, it is still an instructive lower bound for large problem sizes when {\tt IP} is too inefficient to produce a valid optimal solution.

\begin{figure}[p]
\vspace*{-0.25in}
\begin{center}
$\begin{array}{c}
\hspace*{-0.2in}\epsfig{figure=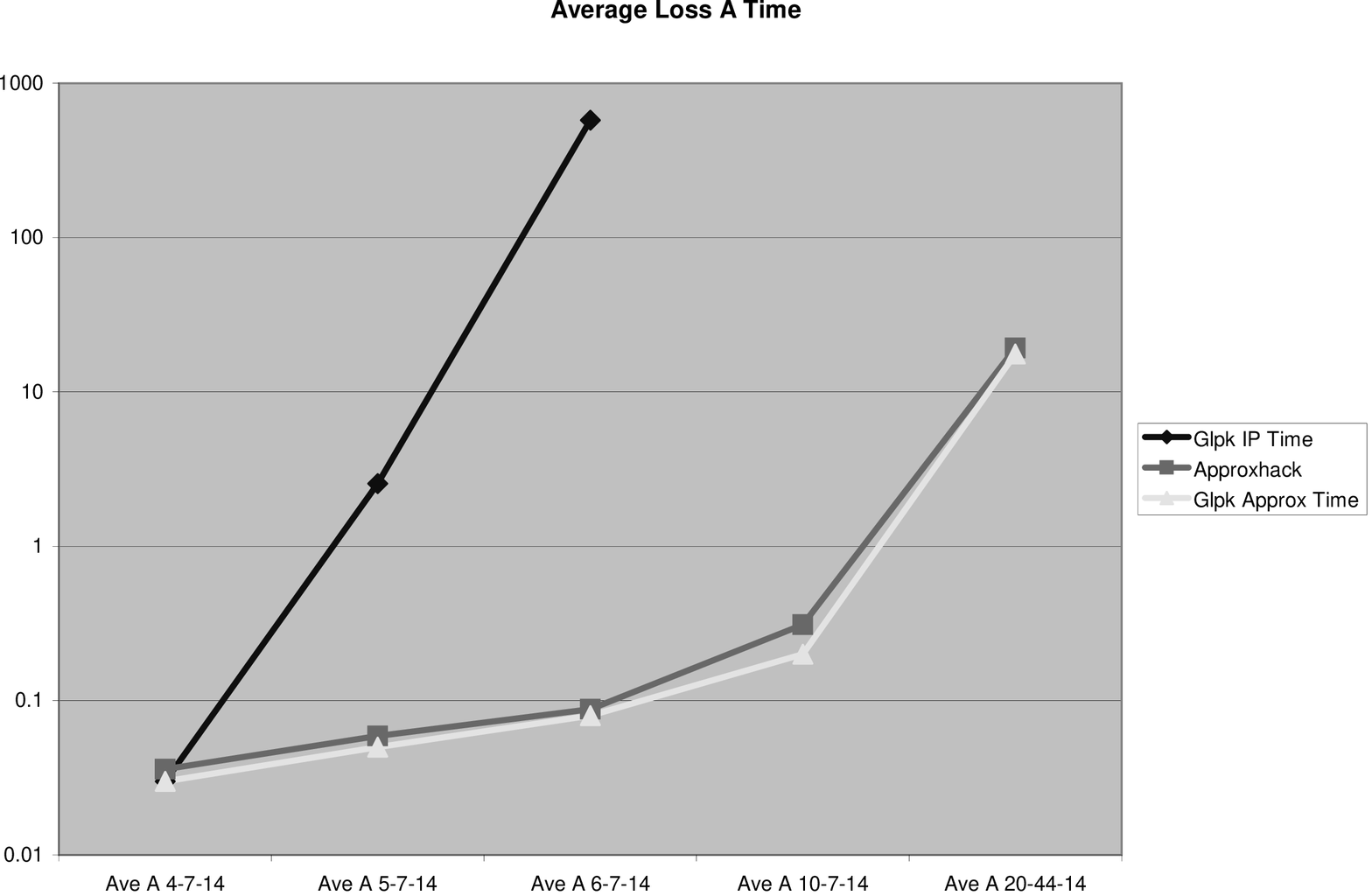,angle=-90,width=3.6in}\vspace*{0in}\\
\epsfig{figure=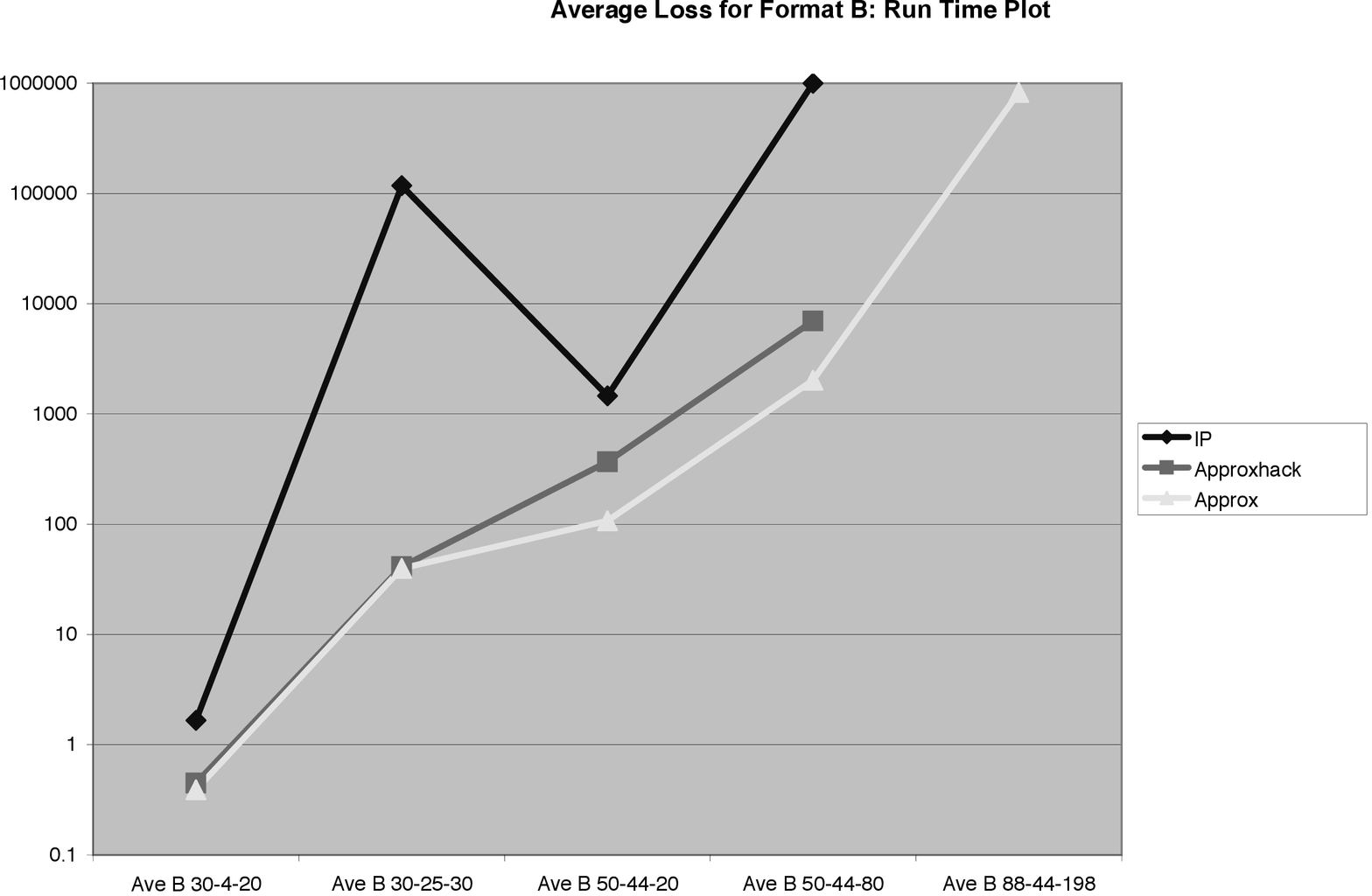,angle=-90,width=4.3in}\vspace*{-1in}\\
\epsfig{figure=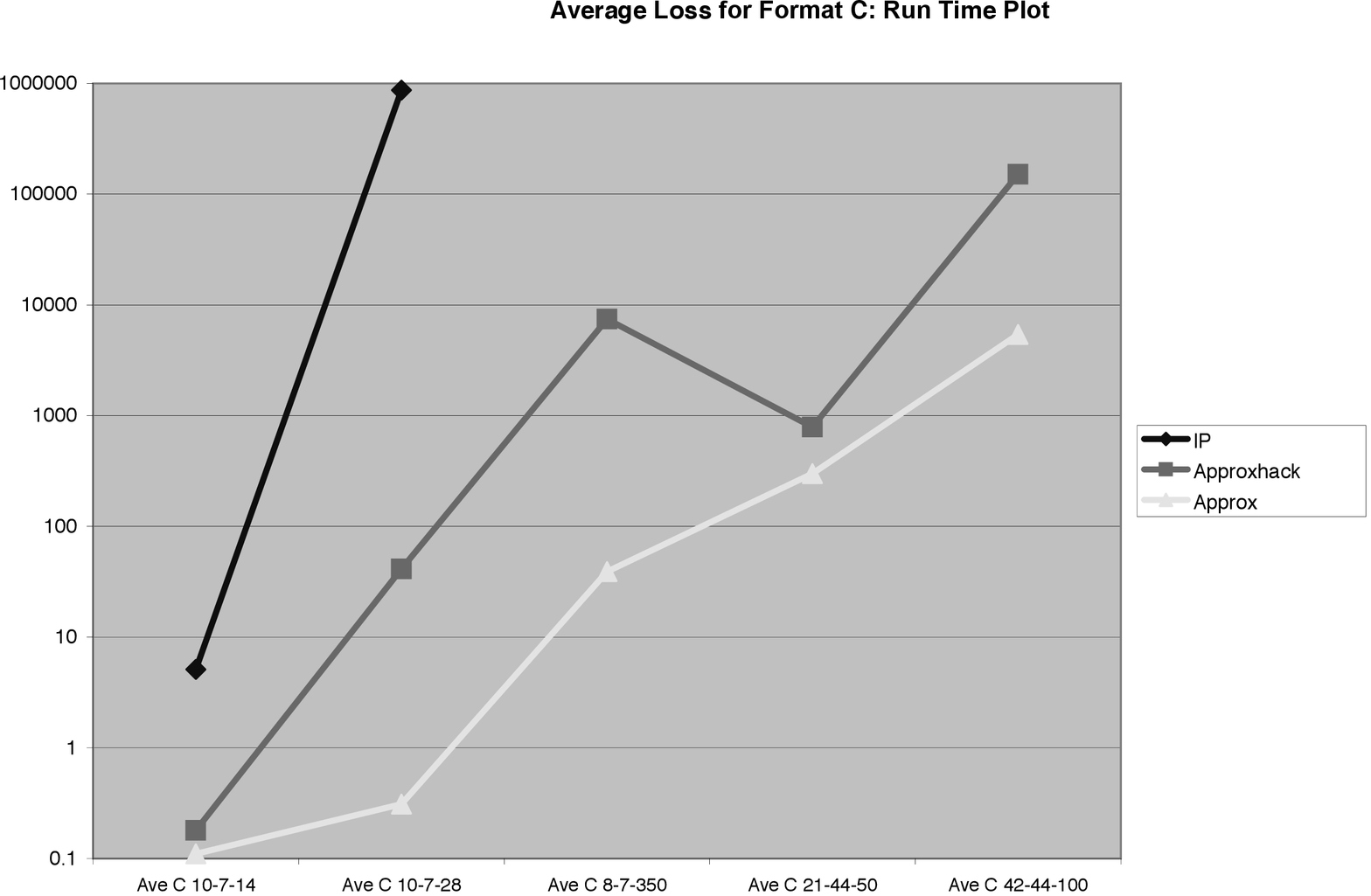,angle=-90,width=4.3in}
\end{array}$
\end{center}
\vspace*{-1in}
\caption{Run time (in seconds) under average loss conditions (log plot)}
\label{AverageLossRunTime}
\end{figure}

\begin{figure}[p]
\vspace*{-0.25in}
\begin{center}
$\begin{array}{c}
\epsfig{figure=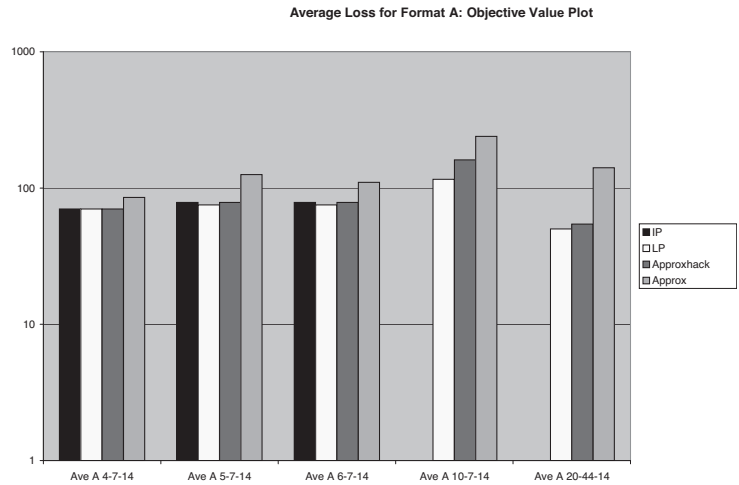,angle=-90,width=3.8in}\vspace*{-0in}\\
\epsfig{figure=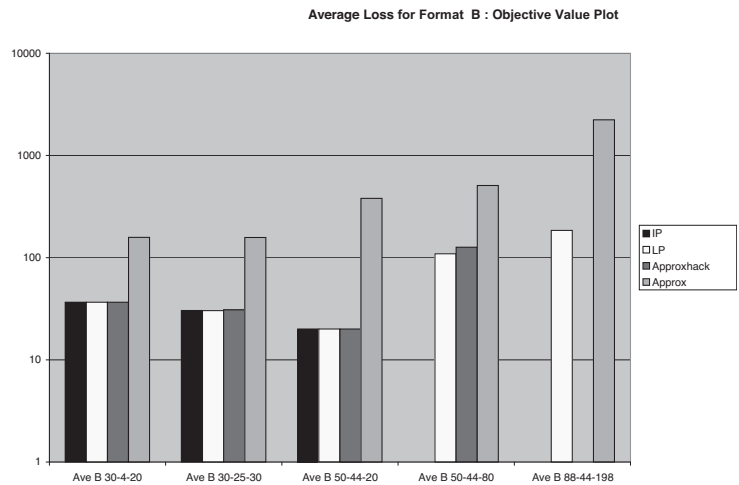,angle=-90,width=3.8in} \vspace*{-0in}\\
\hspace*{1in}\epsfig{figure=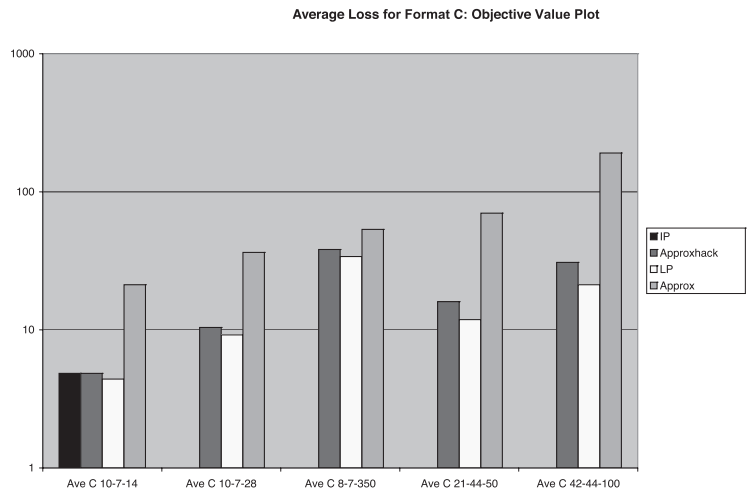,angle=-90,width=4.6in}
\end{array}$
\end{center}
\vspace*{-1in}
\caption{Value of the cost objective function under average loss conditions (log plot)}
\label{AverageLossObjective}
\end{figure}

\begin{figure}[p]
\begin{center}
\vspace*{-0.25in}
$\begin{array}{c}
\epsfig{figure=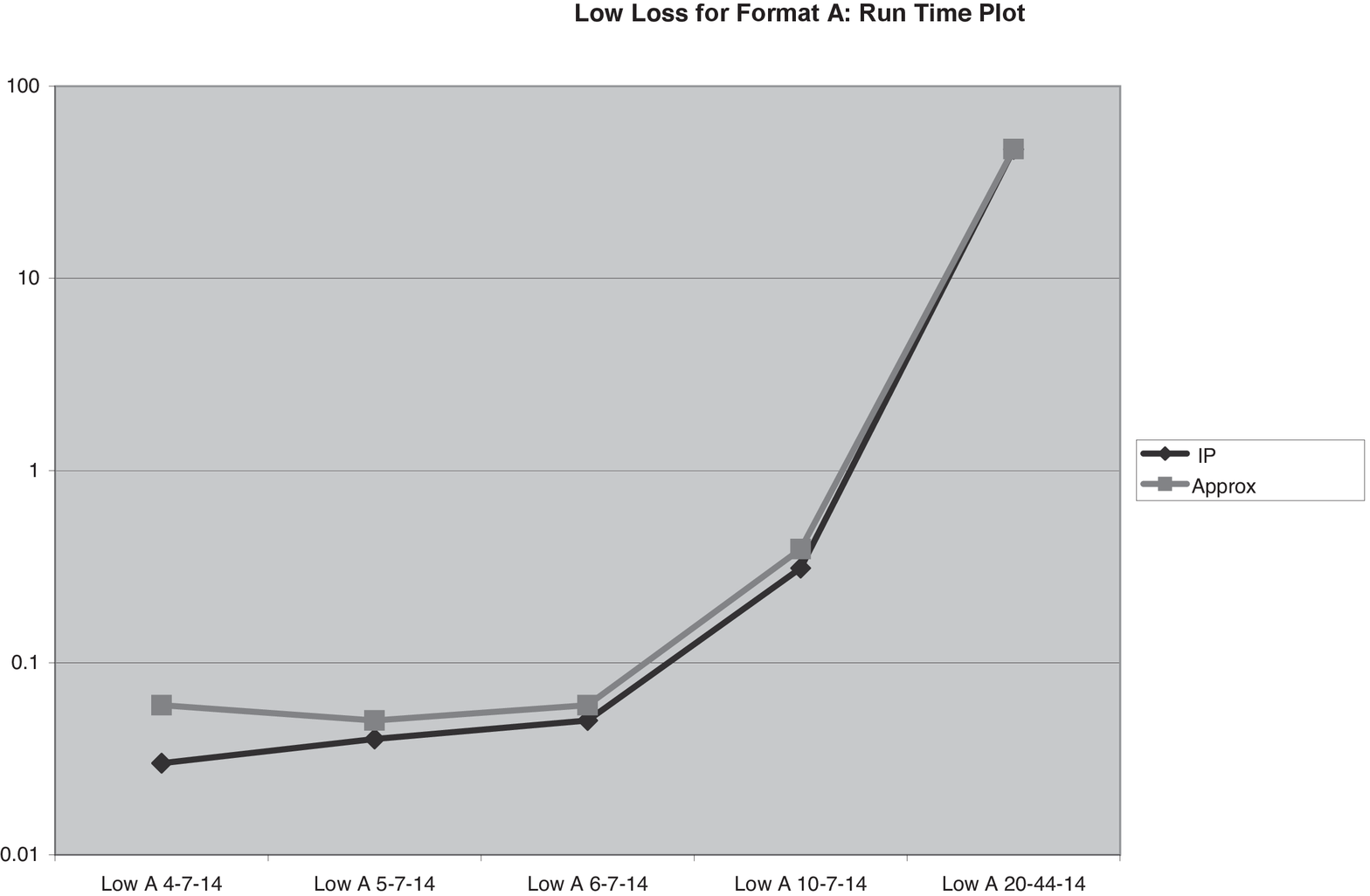,angle=-90,width=3.85in}\vspace*{-.5in} \\
\epsfig{figure=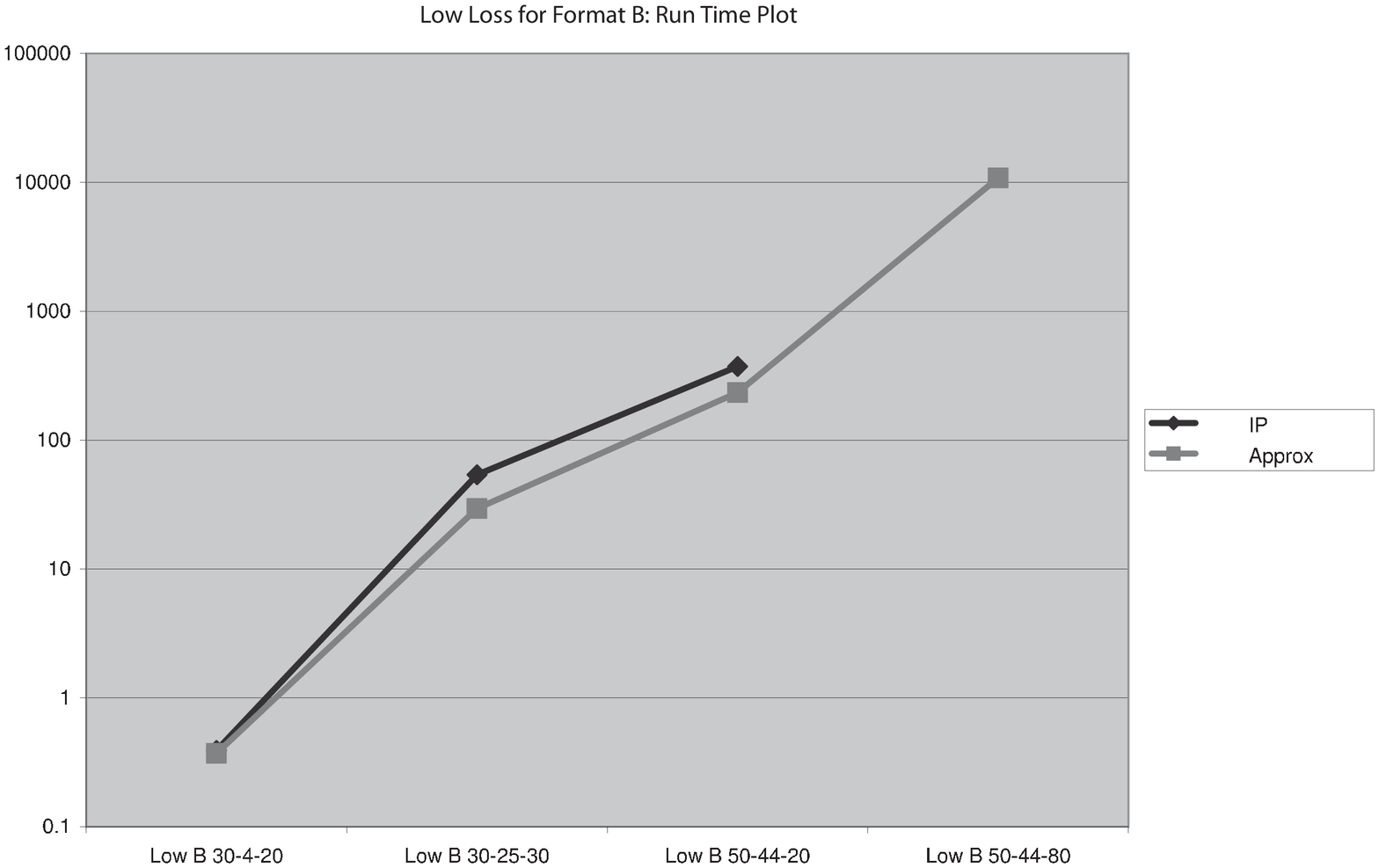,angle=-90,width=4.3in}\vspace*{-.5in} \\
\epsfig{figure=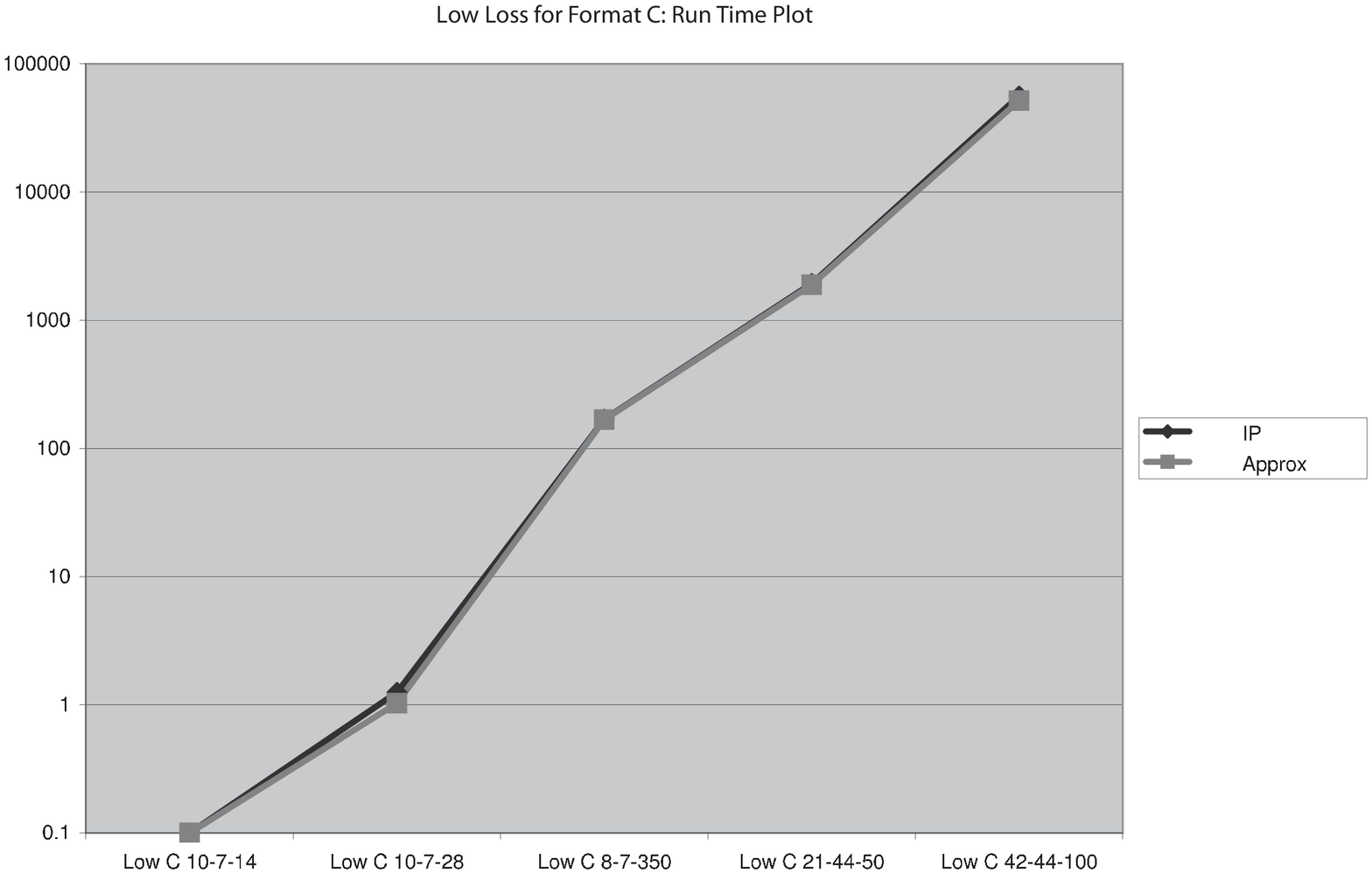,angle=-90,width=4.3in}  
\end{array}$
\end{center}
\caption{Run time (in seconds) under low loss conditions (log plot)}
\label{LowLossRunTime}
\end{figure}

\begin{figure}[p]
\begin{center}
\vspace{-.25in}
$\begin{array}{c} 
\epsfig{figure=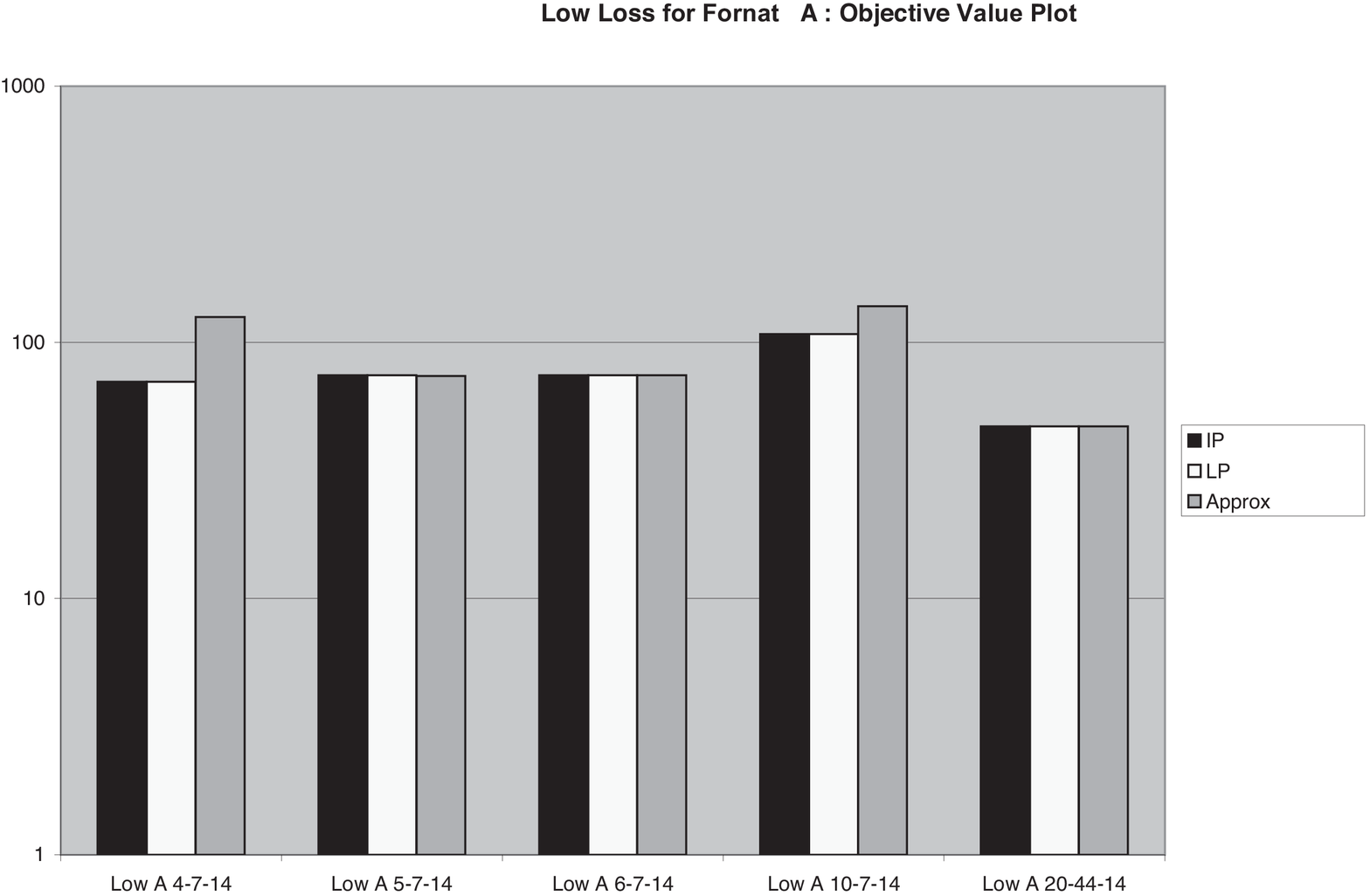,angle=-90,width=3.8in}\\
\epsfig{figure=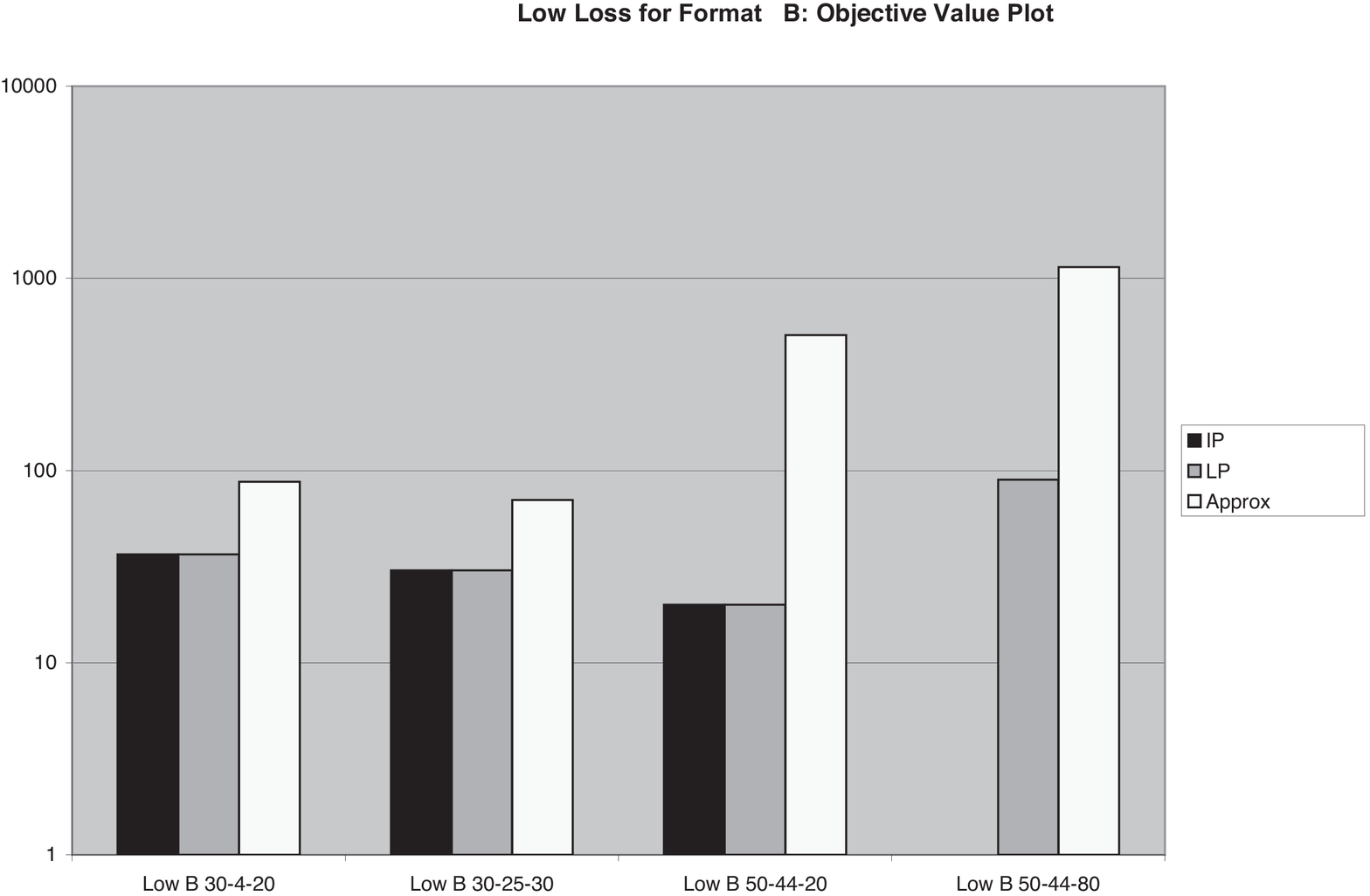,angle=-90,width=3.8in} \\
\epsfig{figure=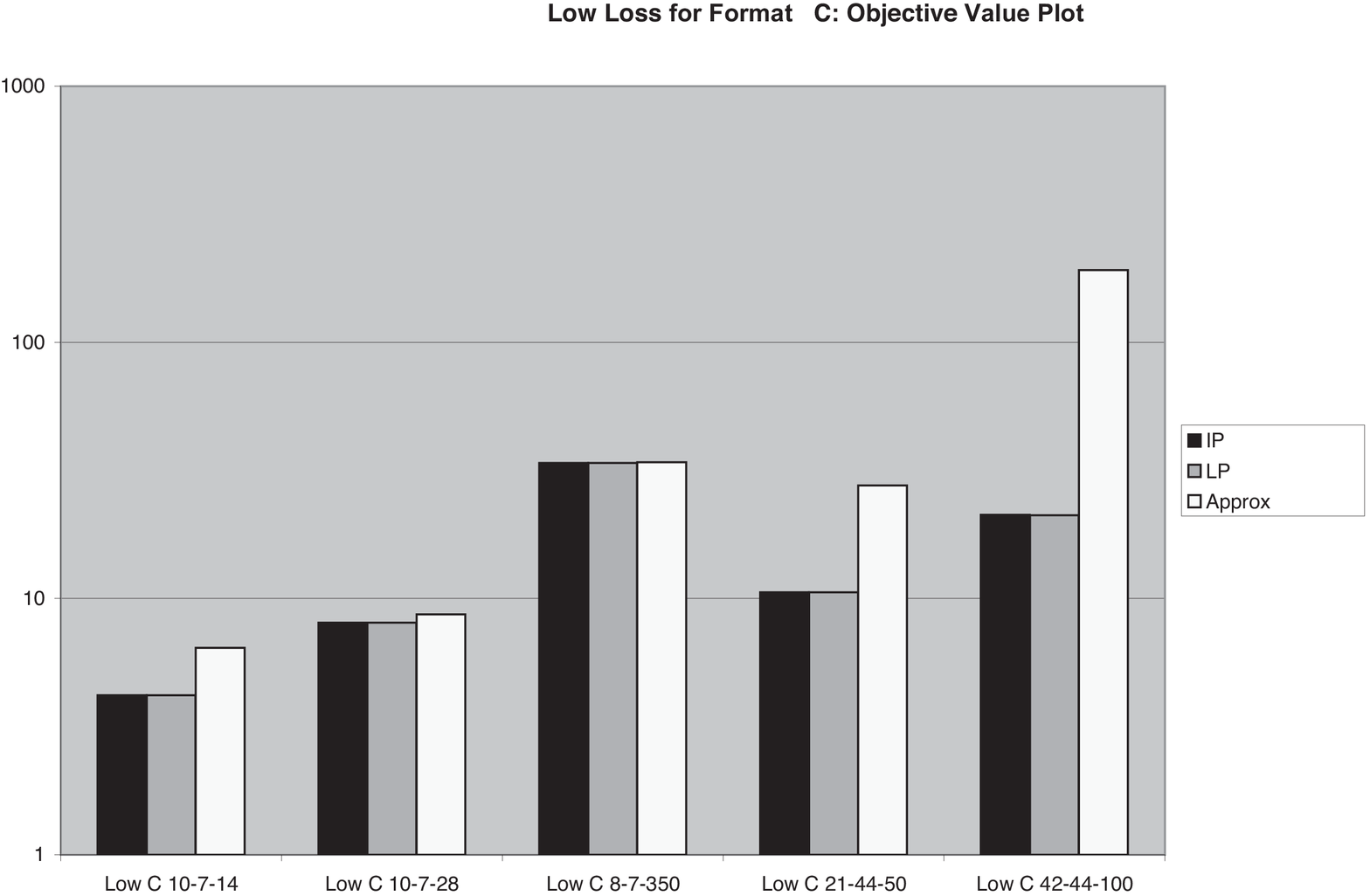,angle=-90,width=3.8in}
\end{array}$
%\vspace{-.5in}
\end{center}
\caption{Value of the cost objective function under low loss conditions (log plot)}
\label{LowLossObjective}
\end{figure}

\begin{figure}[p]
\begin{center}
\vspace*{-.25in}
$\begin{array}{c}
\psfig{figure=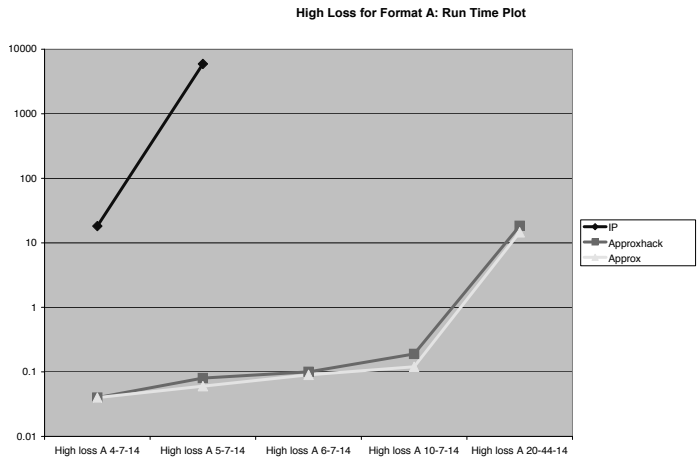,angle=-90,width=4.3in} \vspace*{-.75in}\\
\psfig{figure=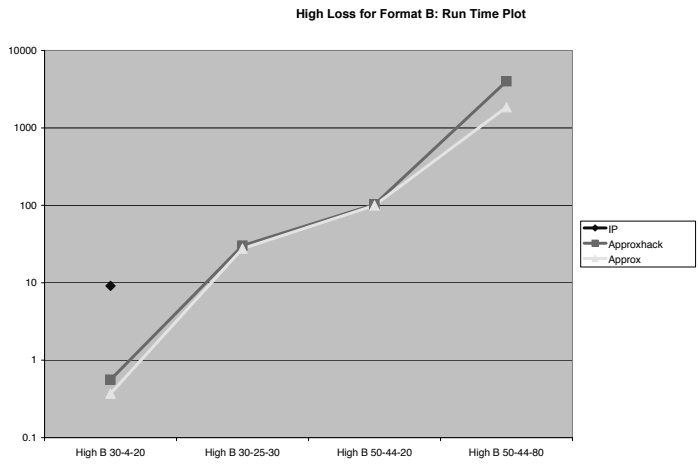,angle=-90,width=4.3in}\vspace*{-.75in}\\
\hspace*{-0.25in}\psfig{figure=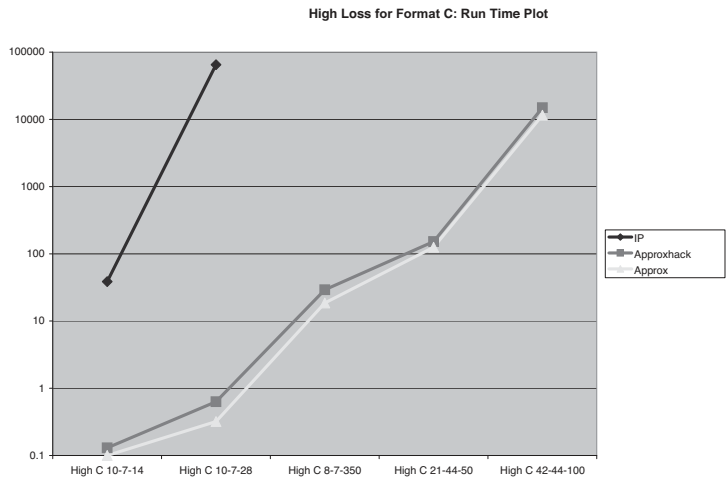,angle=-90,width=3.75in} 
\end{array}$
\end{center}
\caption{Run time (in seconds) under high loss conditions (log plot)}
\label{HighLossRunTime}
\end{figure}

\begin{figure}[p]
\begin{center}
\vspace{-.25in}
$\begin{array}{c}
\psfig{figure=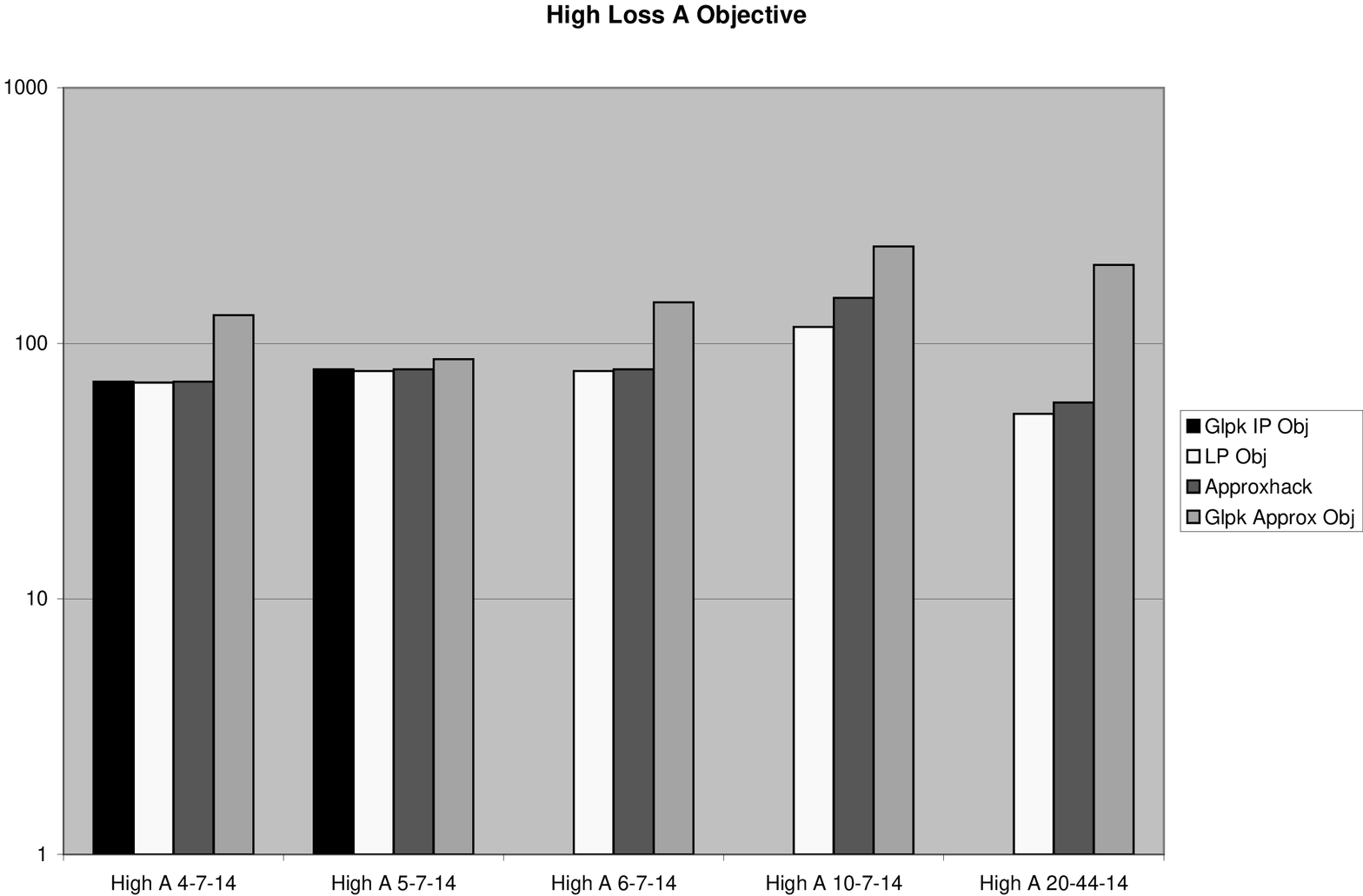,angle=90,width=4.0in} \\
\psfig{figure=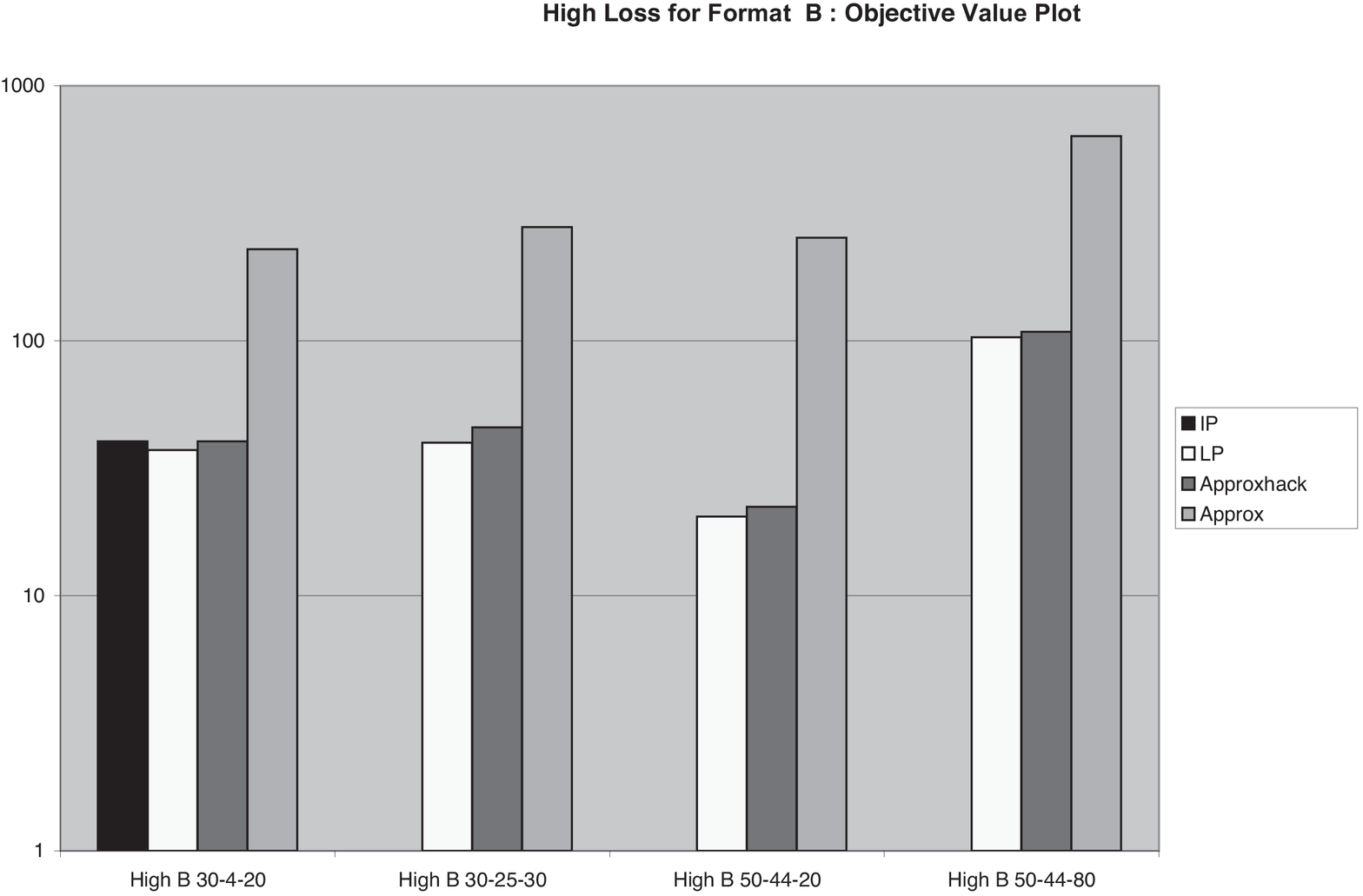,angle=-90,width=3.8in}\\
\psfig{figure=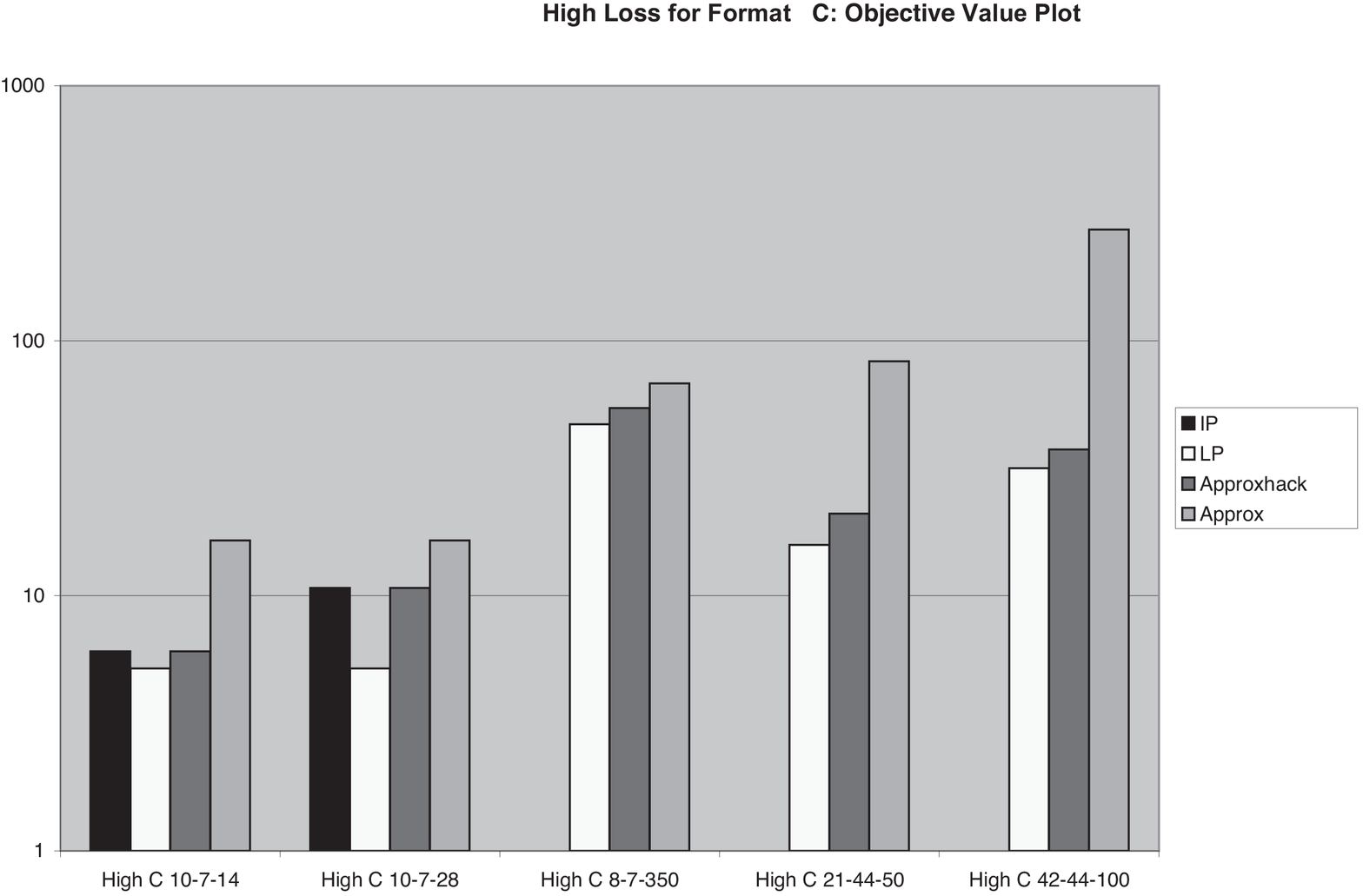,angle=-90,width=3.8in} 
\end{array}$
%\vspace{-.75in}
\end{center}
\caption{Value of the cost objective function under high loss conditions (log plot)}
\label{HighLossObjective}
\end{figure}

\noindent{\bf Average loss conditions.} This is the scenario in which we model moderate loss conditions for the network links by averaging our loss traces from the Akamai network over a calendar day. As expected,  our experimental results show that {\tt Approx} is the fastest in all cases and in some cases the only feasible algorithm (See Figure~\ref{AverageLossRunTime})\footnote{Note that all run time plots in this section use the log scale. Hence, the data points look visually closer than they actually are.}. {\tt IP} is sometimes not feasible beyond even relatively small problem sizes, like the $10\times7\times28$ network for format C that had a problem size of about 2500 variables. On the other hand {\tt Approx} was able to solve a $42\times44\times100$ network in reasonable time
and even a $88\times44\times198$ network that yielded a problem size of about 800,000 variables.  {\tt Approxhack} performed well up to a size 200,000 variables, trailing in time to {\tt Approx}, but failed on higher sizes.

The comparison between the cost objective value achieved by each of the algorithms is also shown in Figure~\ref{AverageLossObjective}.  Note that the optimal cost is achieved by {\tt IP} and can be used as a basis for comparison. For larger network sizes where {\tt IP} was too inefficient to produce a solution, the cost objective value of LP is relevant since it is a lower bound on the best possible cost. As expected, {\tt Approx} produces a larger value for the cost objective function, but is within a reasonable factor of optimal. 

\noindent{\bf Low loss period.}
This is the scenario where we simulate low loss conditions that tend to occur in the non-peak hours and is derived from the Akamai traces. This is the ``easy'' case, since we expect a number of low-loss paths (i.e., high-weight paths) from the sources to the sinks to be available to meet the weight thresholds at the sinks. 
During a low loss period,  {\tt Approx} barely
outperforms {\tt  IP} in terms of time (see Figure~\ref{LowLossRunTime}). Both algorithms run in a feasible amount of time for up to
a $42\times44\times100$ network, which is a problem size of about 200,000 variables. {\tt IP} took too much time and did not, however, complete for the largest network for format B that we simulated, which had dimensions $50\times44\times80$. 

In terms of the value of the cost objective function, when {\tt IP} completed it outperformed {\tt Approx} in most cases
by only a small percentage. But, in a few cases for the larger networks in Formats B and C, {\tt IP} outperformed {\tt Approx} by a larger factor (see Figure~\ref{LowLossObjective}). Since {\tt IP} produced results for all but the largest of networks, we did not run {\tt Approxhack} for the low loss scenario. 

\noindent{\bf High loss period.}
There are periods of time when heavier losses can be experienced across the Internet, typically during the peak hours of the day.  This section addresses this scenario using the peak-hour traces from Akamai.
We see that our algorithm {\tt Approx} is again superior in terms of time to {\tt IP} and also outperforms {\tt Approxhack} (See Figure~\ref{HighLossRunTime}).
{\tt IP} ran for more than 18 hours on the $10\times7\times28$ network for Format C
where {\tt Approx} was able to solve the larger $21\times44\times50$ network for Format C in 3.4 minutes.
{\tt Approx} was even able to solve a relatively dense $42\times44\times100$ network for Format C, which has problem size of 200,000  variables, in 3 hours and 13 minutes.

In terms of the cost objective function that is being minimized, {\tt Approxhack} did beat {\tt Approx}
as one can expect, though in many cases by only a small constant factor. But {\tt Approxhack} took a longer time to produce the better solution. For instance, the $42\times44\times100$ network for Format C took 4 hours 18 minutes (See Figure~\ref{HighLossRunTime}). And, the time difference between {\tt Approx} and {\tt Approxhack} widened even more with larger problem sizes such as the $50\times44\times80$ network of Format B.

\subsection{Experimental Conclusions}

Most of our expectations for running time and quality of the solution were confirmed by our experiments.
Our approximation algorithm {\tt Approx}  was the only feasible way to solve all of the simulated networks and problem sizes in a reasonable amount of time, though it meant sacrificing some solution optimality. As live streaming continues to significantly increase  in popularity each year,   the problem sizes  that we are required to solve increase as well.   Therefore, polynomial-time approximation algorithms such as {\tt Approx} provide the only feasible solutions for the future, as exponential time algorithms such as {\tt IP} and our variant {\tt Approxhack} may not always complete within a reasonable amount of time for the most popular formats. However, our experiments also suggest that the ``easier'' case of low loss is amenable to a more exact solution using {\tt IP}. This would suggest a hybrid approach where {\tt IP} or perhaps even {\tt Approxhack} is used for small and medium-sized networks when the loss is low, while {\tt Approx} is used in all other cases. It is worth noting, however, that overlay network construction is most critical in average and high loss situations where there is a strong need to route streams around hot spots of congestion and packet loss on the Internet. In this regime, {\tt Approx} was the only consistently feasible choice.

\section{Concluding remarks}
\label{sec:concl}
Algorithms for constructing optimum overlay networks are at the heart of modern live stream delivery technology. The algorithms need to be highly efficient, as new overlay networks need to be constructed rapidly in response to the changing failure and loss characteristics of the Internet. Further, the algorithms must scale to even larger networks as global live streaming usage continues to grow. In this work, we provide the first problem formulation and efficient algorithm for constructing live streaming overlays. Besides proving theoretical guarantees,  we have shown that our algorithm is effective on typical real-world inputs. 

A number of interesting challenges for future research remain.  A primary challenge is developing {\em incremental\/} algorithms for solving the overlay network construction problem. Often times there are Internet disruptions that are local to a specific ISP or a specific geography, requiring a quick response by rerouting the impacted streams. In such a situation, it might not be necessary or even feasible to recompute the entire overlay network using an algorithm such as {\tt Approx}. Rather, it would be useful to develop provably-good algorithms that can work in an incremental fashion by recomputing only parts of the overlay network that are directly impacted by the disruption. In such an operational model, the overlay network would be constructed from scratch only when the deployments or the Internet changed in a major way, while smaller changes to the overlay network would be made more frequently in an incremental fashion.

% Appendix
\appendix
%\section*{APPENDIX}
%\setcounter{section}{1}

%Acknowledgments
\section{Acknowledgements}
The authors want to thank Umut Acar, Aditya Akella, Jeff Pang and Maverick Woo for many helpful discussions and for helping collect data for the experiments. A preliminary version of this paper containing only a subset of the results appeared as an extended abstract in the Proceedings of the ACM Symposium on Parallel Algorithms and Architectures (SPAA) in 2003. This work was supported in part by NSF grant CCR-0122581, NSF grant CCR-012258, NSF Career award No. CCR-97-03017 and NSF Award CNS-05-19894.

%Bibliography
\bibliographystyle{alpha}
\bibliography{reflector_total}

\newcommand{\etalchar}[1]{$^{#1}$}
\begin{thebibliography}{KSW{\etalchar{+}}04}

\bibitem[AMO93]{AhujaMO93}
Ravindra~K. Ahuja, Thomas~L. Magnanti, and James~B. Orlin.
\newblock {\em Network Flows: Theory, Algorithms, and Applications}.
\newblock Prentice Hall, 1993.

\bibitem[APM{\etalchar{+}}04]{APMSS04}
A.~Akella, J.~Pang, B.~Maggs, S.~Seshan, and A.~Shaikh.
\newblock {A comparison of overlay routing and multihoming route control}.
\newblock {\em ACM SIGCOMM Computer Communication Review}, 34(4):93--106, 2004.

\bibitem[ASV06]{AdlerSV06}
M.~Adler, R.~Sitaraman, and H.~Venkataramani.
\newblock Algorithms for optimizing bandwidth costs on the {Internet}.
\newblock {\em IEEE Workshop on Hot Topics in Web Systems and Technologies
  (HOTWEB)}, pages 1--9, 2006.

\bibitem[Bel10]{Belson10}
David Belson.
\newblock Akamai state of the {Internet} report, {Q4} 2009.
\newblock {\em SIGOPS Operating Systems Review}, 44:27--37, August 2010.

\bibitem[BR01]{BaevR01}
Ivan~D. Baev and Rajmohan Rajaraman.
\newblock Approximation algorithms for data placement in arbitrary networks.
\newblock In {\em Proceedings of the Twelfth Annual ACM-SIAM Symposium on
  Discrete Algorithms}, pages 661--670, Philadelphia, PA, USA, January 2001.
  Society for Industrial and Applied Mathematics.

\bibitem[CG99]{CharikarG99}
Moses Charikar and Sudipto Guha.
\newblock Improved combinatorial algorithms for the facility location and
  k-median problems.
\newblock In {\em Proceedings of the 40th Annual Symposium on Foundations of
  Computer Science}, pages 378--388, Washington, DC, USA, October 1999. IEEE
  Computer Society.

\bibitem[Chu98]{Chudak98}
F.A. Chudak.
\newblock Improved algorithms for uncapacitated facility location problem.
\newblock In {\em Proceedings of the 6th Conference on Integer Programming and
  Combinatorial Optimization}, pages 180--194. Springer, 1998.

\bibitem[Chv79]{Chvatal79}
V.~Chvatal.
\newblock {A greedy heuristic for the set-covering problem}.
\newblock {\em Mathematics of Operations Research}, 4(3):233--235, 1979.

\bibitem[CLRS09]{CormenLRS2009}
T.~H. Cormen, C.~E. Leiserson, R.~L. Rivest, and C.~Stein.
\newblock {\em {Introduction to algorithms}}.
\newblock The MIT Press, 3 edition, 2009.

\bibitem[CRSZ02]{Chu00}
Y.~Chu, S.G. Rao, S.~Seshan, and H.~Zhang.
\newblock {A case for end system multicast}.
\newblock {\em , IEEE Journal on Selected Areas in Communications},
  20(8):1456--1471, 2002.

\bibitem[Dee91]{Deering91}
S.E Deering.
\newblock {\em {Multicast routing in a datagram Internetwork}}.
\newblock Ph.D Thesis, Dept. of Computer Science, Stanford University, December
  1991.

\bibitem[DMP{\etalchar{+}}02]{DilleyMPPSW02}
John Dilley, Bruce~M. Maggs, Jay Parikh, Harald Prokop, Ramesh~K. Sitaraman,
  and William~E. Weihl.
\newblock Globally distributed content delivery.
\newblock {\em IEEE Internet Computing}, 6(5):50--58, 2002.

\bibitem[Eri94]{Eriksson94}
H.~Eriksson.
\newblock {Mbone: The multicast backbone}.
\newblock {\em Communications of the ACM}, 37(8):54--60, 1994.

\bibitem[FABK03]{FABK03}
N.~Feamster, D.G. Andersen, H.~Balakrishnan, and M.F. Kaashoek.
\newblock {Measuring the effects of Internet path faults on reactive routing}.
\newblock In {\em Proceedings of the 2003 ACM SIGMETRICS International
  Conference on Measurement and Modeling of Computer Systems}, pages 126--137.
  ACM, 2003.

\bibitem[Fei98]{Feige98}
U.~Feige.
\newblock {A threshold of ln n for approximating set cover}.
\newblock {\em Journal of the ACM}, 45(4):634--652, 1998.

\bibitem[GJ79]{GJ79}
M.~R. Garey and D.~S. Johnson.
\newblock {\em Computers and Intractability: A Guide to the Theory of
  NP-completeness}.
\newblock W.H. Freeman and Co, San Francisco, CA, 1979.

\bibitem[GK98]{GuhaK98}
Sudipto Guha and Samir Khuller.
\newblock Greedy strikes back: improved facility location algorithms.
\newblock In {\em Proceedings of the Ninth Annual ACM-SIAM Symposium on
  Discrete Algorithms}, pages 649--657, Philadelphia, PA, USA, January 1998.
  SIAM.

\bibitem[GM02]{GuhaM02}
S.~Guha and K.~Munagala.
\newblock {Improved algorithms for the data placement problem}.
\newblock In {\em Proceedings of the Thirteenth Annual ACM-SIAM symposium on
  Discrete Algorithms}, pages 106--107. SIAM, January 2002.

\bibitem[GMM01]{GuhaMM01}
S.~Guha, A.~Meyerson, and K.~Munagala.
\newblock {Improved algorithms for fault tolerant facility location}.
\newblock In {\em Proceedings of the Twelfth Annual ACM-SIAM Symposium on
  Discrete Algorithms}, pages 636--641. SIAM, January 2001.

\bibitem[Hoc82]{Hochbaum82}
D.S. Hochbaum.
\newblock {Heuristics for the fixed cost median problem}.
\newblock {\em Mathematical Programming}, 22(1):148--162, 1982.

\bibitem[Hoe63]{Hoeffding63}
W.~Hoeffding.
\newblock {Probability inequalities for sums of bounded random variables}.
\newblock {\em Journal of the American Statistical Association},
  58(301):13--30, 1963.

\bibitem[JMS02]{JainMS02}
K.~Jain, M.~Mahdian, and A.~Saberi.
\newblock {A new greedy approach for facility location problems}.
\newblock In {\em Proceedings of the Thirty-Fourth Annual ACM Symposium on
  Theory of Computing}, pages 731--740. ACM, May 2002.

\bibitem[Joh74]{Johnson74}
D.S. Johnson.
\newblock {Approximation algorithms for combinatorial problems*}.
\newblock {\em Journal of Computer and System Sciences}, 9(3):256--278, 1974.

\bibitem[JV99]{JainV99}
Kamal Jain and V.~Vazirani.
\newblock Primal-dual approximation algorithms for metric facility location and
  k-median problems.
\newblock In {\em Proceedings of the 40th Annual Symposium on Foundations of
  Computer Science}, pages 2--13, Washington, DC, USA, October 1999. IEEE.

\bibitem[JV04]{JainV00}
K.~Jain and V.V. Vazirani.
\newblock {An approximation algorithm for the fault tolerant metric facility
  location problem}.
\newblock {\em Algorithmica}, 38(3):433--439, 2004.

\bibitem[Kar95]{Karger99}
D.R. Karger.
\newblock {A randomized fully polynomial time approximation scheme for the all
  terminal network reliability problem}.
\newblock In {\em Proceedings of the Twenty-Seventh Annual ACM Symposium on
  Theory of Computing}, pages 11--17. ACM, May 1995.

\bibitem[KLR{\etalchar{+}}87]{KLRTVV87}
M.~Karp, T.~Leighton, R.~Rivest, C.~Thompson, U~Vazirani, and V.~Vazirani.
\newblock Global wire routing in two-dimensional arrays.
\newblock {\em Algorithmica, 2}, pages 113--129, 1987.

\bibitem[KPR99]{KorupoluPR99}
M.~Korupolu, G.~Plaxton, and R.~Rajaraman.
\newblock Placement algorithms for hierarchical cooperative caching.
\newblock {\em Proceedings of the 10th ACM-SIAM Symposium on Discrete
  Algorithms}, pages 586--595, January 1999.

\bibitem[KSW{\etalchar{+}}04]{KontothanasisSWHKMSS04}
L.~Kontothanassis, R.~Sitaraman, J.~Wein, D.~Hong, R.~Kleinberg, B.~Mancuso,
  D.~Shaw, and D.~Stodolsky.
\newblock {A transport layer for live streaming in a content delivery network}.
\newblock {\em Proceedings of the IEEE}, 92(9):1408--1419, 2004.

\bibitem[LRLZ08]{LiuRLZ08}
J.~Liu, S.G. Rao, B.~Li, and H.~Zhang.
\newblock {Opportunities and challenges of peer-to-peer internet video
  broadcast}.
\newblock {\em Proceedings of the IEEE}, 96(1):11--24, 2008.

\bibitem[LY94]{LundY94}
C.~Lund and M.~Yannakakis.
\newblock {On the hardness of approximating minimization problems}.
\newblock {\em Journal of the ACM}, 41(5):960--981, 1994.

\bibitem[MMP01]{MeyersonMP01}
Adam Meyerson, Kamesh Munagala, and Serge Plotkin.
\newblock Web caching using access statistics.
\newblock In {\em Proceedings of the Twelfth Annual ACM-SIAM Symposium on
  Discrete Algorithms}, pages 354--363, Philadelphia, PA, USA, January 2001.
  SIAM.

\bibitem[MR95]{MotwaniR95}
R.~Motwani and P.~Raghavan.
\newblock {\em {Randomized Algorithms}}, chapter~4, pages 67--74.
\newblock Cambridge University Press, 1995.

\bibitem[MYZ02]{Mahdian}
M.~Mahdian, Y.~Ye, and J.~Zhang.
\newblock {Improved approximation algorithms for metric facility location
  problems}.
\newblock In {\em Approximation Algorithms for Combinatorial Optimization},
  pages 229--242. Springer, 2002.

\bibitem[NSS10]{NygrenSS10}
Erik Nygren, Ramesh~K. Sitaraman, and Jennifer Sun.
\newblock {The Akamai network: a platform for high-performance Internet
  applications}.
\newblock {\em SIGOPS Operating Systems Review}, 44:2--19, August 2010.

\bibitem[PB83]{ProvanB83}
J.~Scott Provan and Michael~O. Ball.
\newblock The complexity of counting cuts and of computing the probability that
  a graph is connected.
\newblock {\em SIAM Journal on Computing}, 12(4):777--788, 1983.

\bibitem[PS02]{PadmanabhanS02}
V.~N. Padmanabhan and K.~Sripanidkulchai.
\newblock The case for cooperative networking.
\newblock In {\em Revised Papers from the First International Workshop on
  Peer-to-Peer Systems}, pages 178--190, London, UK, 2002. Springer-Verlag.

\bibitem[PTW01]{PalTW01}
M.~P\'{a}l, \'{E}. Tardos, and T.~Wexler.
\newblock Facility location with nonuniform hard capacities.
\newblock In {\em Proceedings of the Forty-Second IEEE Symposium on Foundations
  of Computer Science}, pages 329--338, Washington, DC, USA, October 2001. IEEE
  Computer Society.

\bibitem[PWCS02]{Ven01}
V.N. Padmanabhan, H.J. Wang, P.A. Chou, and K.~Sripanidkulchai.
\newblock {Distributing streaming media content using cooperative networking}.
\newblock In {\em Proceedings of the Twelfth International Workshop on Network
  and Operating Systems Support for Digital Audio and Video}, pages 177--186.
  ACM, 2002.

\bibitem[SMZ04]{sripanidkulchai2004analysis}
K.~Sripanidkulchai, B.~Maggs, and H.~Zhang.
\newblock {An analysis of live streaming workloads on the Internet}.
\newblock In {\em Proceedings of the 4th ACM Internet Measurement Conference
  (IMC)}, pages 41--54. ACM, October 2004.

\bibitem[ST93]{ShmoysT93}
D.B. Shmoys and {\'E}.~Tardos.
\newblock {An approximation algorithm for the generalized assignment problem}.
\newblock {\em Mathematical Programming}, 62(1):461--474, 1993.

\bibitem[ST01]{SriTeo01}
A.~Srinivasan and C.~Teo.
\newblock {A constant-factor approximation algorithm for packet routing and
  balancing vs. global criteria}.
\newblock {\em SIAM Journal of Computing}, 30(6):2051--2068, 2001.

\bibitem[STA97]{ShmoysTA97}
D.B. Shmoys, {\'E}.~Tardos, and K.~Aardal.
\newblock {Approximation algorithms for facility location problems}.
\newblock In {\em Proceedings of the Twenty-Ninth Annual ACM Symposium on
  Theory of Computing}, pages 265--274. ACM, May 1997.

\bibitem[Svi02]{Svir01}
Maxim Sviridenko.
\newblock An improved approximation algorithm for the metric uncapacitated
  facility location problem.
\newblock In {\em Proceedings of the 9th International Conference on Integer
  Programming and Combinatorial Optimization}, pages 240--257, London, UK,
  2002. Springer-Verlag.

\bibitem[Val79]{Valiant79}
L.G. Valiant.
\newblock {The complexity of enumeration and reliability problems}.
\newblock {\em SIAM Journal on Computing}, 8:410--421, 1979.

\end{thebibliography}
\end{document}